\documentclass[11pt,letterpaper]{article}
\usepackage[margin=1in]{geometry}

\usepackage{cite}

\usepackage[bookmarks=true,pdfstartview=FitH,colorlinks,linkcolor=darkred,filecolor=darkred,citecolor=darkred,urlcolor=darkred]{hyperref}

\usepackage{colortbl}
\usepackage{mdframed}

\usepackage{paralist}
\usepackage{pdfpages}
\usepackage{enumitem}
\usepackage{float}
\usepackage{booktabs}
\usepackage[override]{cmtt}
\usepackage[small]{titlesec}

\titlespacing*{\section}
{0pt}{2.5ex plus 1ex minus .2ex}{1ex plus .2ex}
\titlespacing*{\subsection}
{0pt}{2.5ex plus 1ex minus .2ex}{1ex plus .2ex}

\usepackage{color,xcolor}
\usepackage{graphicx,color,eso-pic}
\usepackage{amsmath,amssymb,stmaryrd}
\usepackage[
	lambda,
	advantage,
	operators,
	sets,
	adversary,
	landau,
	probability,
	notions,
	logic,
	ff,
	mm,
	primitives,
	events,
	complexity,
	asymptotics,
	keys]
{cryptocode}

\usepackage{tikz}
\usetikzlibrary{fit, shapes, positioning, patterns, arrows,shadows}

\tikzset{
    master/.style={
        execute at end picture={
            \coordinate (lower right) at (current bounding box.south east);
            \coordinate (upper left) at (current bounding box.north west);
        }
    },
    slave/.style={
        execute at end picture={
            \pgfresetboundingbox
            \path (upper left) rectangle (lower right);
        }
    }
}

\newcommand{\vone}{\vspace{.1in}}

\newcommand{\hide}[1]{}

\newcommand{\boi}{\begin{itemize}}
\newcommand{\eoi}{\end{itemize}}
\newcommand{\bii}{\begin{itemize}}
\newcommand{\eii}{\end{itemize}}

\newcommand{\bei}{\begin{enumerate}}
\newcommand{\eei}{\end{enumerate}}

\newcommand{\bins}{\beta}

\renewcommand{\highlightkeyword}[2][\ ]{\ensuremath{\textrm{#2}}#1}

\newcommand{\msf}[1]{\ensuremath{{\mathsf {#1}}}}
\newcommand{\mcal}[1]{\ensuremath{\mathcal {#1}}}

\newcommand{\algA}{{\ensuremath{\mcal{A}}}\xspace}

\newcommand{\Sim}{\ensuremath{\msf{Sim}}\xspace}

\newcommand{\E}{\mathbb{E}}

\newcommand{\group}{\ensuremath{g}\xspace}

\definecolor{darkgreen}{rgb}{0,0.5,0}
\definecolor{lightblue}{RGB}{0,176,240}
\definecolor{darkblue}{RGB}{0,112,192}
\definecolor{lightpurple}{RGB}{124, 66, 168}
\definecolor{grey}{RGB}{139, 137, 137}
\definecolor{maroon}{RGB}{178, 34, 34}
\definecolor{green}{RGB}{34, 139, 34}
\definecolor{types}{RGB}{72, 61, 139}
\definecolor{gold}{rgb}{0.8, 0.33, 0.0}

\definecolor{darkgray}{gray}{0.3}

\newcommand{\skiptext}[1]{}

\usepackage{boxedminipage}
\usepackage{cite}
\usepackage{times}
\usepackage{url}
\usepackage{graphicx}
\usepackage[bf]{caption}
\usepackage[justification=centering]{subcaption}
\usepackage{color}
\usepackage{xspace}
\usepackage{multirow}
\usepackage{amsmath,amsthm,amstext,amssymb,amsfonts,latexsym}
\usepackage{wrapfig}

\usepackage{algorithm}
\usepackage[noend]{algpseudocode}
\usepackage{algorithmicx}

\definecolor{darkred}{rgb}{0.5, 0, 0}
\definecolor{darkgreen}{rgb}{0, 0.5, 0}
\definecolor{darkblue}{rgb}{0,0,0.5}

\newcommand\markx[2]{}

\renewcommand{\path}{\ensuremath{\mathsf{path}}\xspace}

\newcommand{\addr}{{\ensuremath{\mathsf{addr}}}\xspace}

\renewcommand{\negl}{{\sf negl}}

\newcommand{\N}{\mathbb{N}}

\newcommand{\ignore}[1]{}

\renewcommand{\paragraph}[1]{\vspace{5pt}\noindent\textbf{#1}}

\newcounter{task}

{
\theoremstyle{definition}
\newtheorem{definition}{Definition}
\newtheorem{remark}{Remark}
}

\newtheorem{thm}{Theorem}[section]      

\newtheorem{theorem}[thm]{Theorem}

\newtheorem{lemma}[thm]{Lemma}

\newtheorem{fact}[thm]{Fact}

\newtheoremstyle{boxes}
{2pt}
{0pt}
{}
{}
{\bfseries}
{}
{\newline}
{\thmname{#1}\thmnumber{ #2}:  
\thmnote{#3}}

\theoremstyle{boxes}

\newcommand{\elaine}[1]{{\footnotesize\color{magenta}[Elaine: #1]}}
\newcommand{\vlr}[1]{{\footnotesize\color{orange}[Vijaya: #1]}}
\newcommand{\VLR}[1]{{\footnotesize\color{orange}[Vijaya: #1]}}
\newcommand{\gnote}[1]{{\footnotesize\color{blue}[Gilad: #1]}}
\newcommand{\weikai}[1]{{\footnotesize\color{green}[WK: #1]}}

\renewcommand{\gnote}[1]{}
\renewcommand{\weikai}[1]{}

\renewcommand{\elaine}[1]{}
\renewcommand{\VLR}[1]{}
\renewcommand{\vlr}[1]{}

\newcounter{cnt:challenge}

\begin{document}
\begin{titlepage}
\title{Data Oblivious Algorithms for Multicores\thanks{This work was supported in part by NSF grant CCF-2008241, CNS-2128519, and CNS-2044679. The extended abstract of this paper appears in {\it Proc. ACM Symp. on Parallelism 
in Algorithms and Architectures (SPAA)}, 2021.}}
\author{Vijaya Ramachandran \\ UT Austin \\ {\tt vlr@cs.utexas.edu}  
\and Elaine Shi \\ CMU \\ {\tt runting@gmail.com}}
\date{}

\maketitle

\begin{abstract}
A data-oblivious algorithm is an algorithm whose memory access pattern is independent of the input values.
As secure processors such as Intel SGX (with hyperthreading) become widely adopted, there is a growing appetite for private analytics on big data. Most prior works on data-oblivious algorithms adopt the classical PRAM model to capture parallelism. However, it is widely understood that PRAM does not best capture realistic multicore processors, nor does it reflect parallel programming models adopted in practice.

We initiate the study of parallel data oblivious algorithms on realistic multicores, best captured by the binary fork-join model of computation. We
present a data-oblivious 
CREW binary fork-join sorting
algorithm with optimal total work and optimal (cache-oblivious) cache complexity, and in $O(\log n \log \log n)$ span (i.e., parallel time); these bounds match the best-known bounds for binary fork-join cache-efficient
 insecure algorithms.  Using our sorting algorithm as a core primitive, we show how to data-obliviously simulate general PRAM algorithms in the binary fork-join model with non-trivial efficiency, and we present 
 data-oblivious algorithms for several applications including list ranking, Euler tour, tree contraction, connected components, and minimum spanning forest. All of our data oblivious algorithms have bounds that either
 match or improve over the best known bounds for insecure algorithms.  

Complementing these asymptotically efficient results, we present a practical variant of our sorting algorithm that is self-contained and potentially implementable.
It has optimal caching cost, and it is 
only a $\log \log n$  factor 
off from optimal work and about a $\log n$ factor off in terms of span;
moreover, it achieves small constant factors in its bounds. 
We also present
an EREW variant with optimal work and
caching cost, and with the same asymptotic
span.
\end{abstract}

\end{titlepage}

\section{Introduction}
\label{sec:intro}

As secure processors such as Intel SGX (with hyperthreading) become widely adopted, there is a growing appetite for private analytics on big data. Most prior works on data-oblivious algorithms adopt the classical PRAM model to capture parallelism. However, it is widely understood that PRAM does not best capture realistic multicore processors, nor does it reflect parallel programming models adopted in practice.

We initiate the study of {\it data-oblivious} algorithms for a  multicore architecture 
where parallelism 
and synchronization 
are expressed with nested 
 binary fork-join operations.
Imagine that a client outsources encrypted data to an untrusted 
cloud server which is equipped with a secure, multicore processor architecture 
(e.g., Intel SGX with hyperthreading).
All data contents are encrypted to the secure processor's 
secret key 
both at rest and in transit.
Data is decrypted only inside the secure cores' hardware sandboxes 
where computation takes place.
However, it is well-known that
encryption alone 
does not guarantee privacy, since access patterns to even encrypted
data leak a lot of sensitive information~\cite{accesspatternleak,controlchannel}.
To defend against access
pattern leakage, an 
active
  line 
of work~\cite{oram00,oram10,oblivm,circuitopram,opram,asiacrypt11}
has focused on how to design algorithms whose access pattern 
distributions do not depend on the secret inputs --- such algorithms are
called {\it data-oblivious} algorithms, and are the focus of our work.
In this paper we present nested fork-join data-oblivious algorithms for several fundamental 
problems that are highly parallel, and work- and cache-efficient.
Throughout this paper, we consider only data oblivious algorithms
that are unconditionally secure (often called
``statistically secure''), i.e., without the need to make any computational
hardness assumptions.

There has been  some prior work exploring the design of {\it parallel} data-oblivious algorithms.
Most of these prior parallel data-oblivious algorithms~\cite{graphsc,opram,circuitopram}
adopted PRAM as the model of computation.
However, the global synchronization between cores at each parallel step of the computation in a PRAM algorithm
does not best capture modern multicore architectures, where the cores typically 
proceed asynchronously.
To better account for synchronization cost, a long 
line of work~\cite{bfork01,lowdepthco,bfork02,bfork03,bfork04,CR07,CR08,FS09,CR12,CR12b,CR13,CRSB13,CR17,CR17b,bfork05,bfork06,bfork07,bfork08} 
has adopted a multithreaded computation model  with CREW (Concurrent Read Exclusive Write) shared memory
in which parallelism is expressed through paired fork and join operations.
A binary fork spawns two tasks that can
execute in parallel. Its corresponding join is a synchronization point: both of the spawned
tasks must complete before the computation can proceed beyond this join.
Such a binary fork-join model is also  
adopted in practice, and supported by programming
systems such as Cilk~\cite{cilk5}, the 
Java fork-join framework~\cite{jforkjoin}, X10~\cite{x10}, 
Habanero~\cite{habanero}, Intel Threading Building Blocks~\cite{tbb},
and the Microsoft Task Parallel Library~\cite{tpl}.

Efficiency in  a multi-threaded algorithm with binary fork-joins is measured through the following metrics:
the algorithm's {\it total work} (i.e., the sequential execution time), 
its {\it cache complexity} (i.e., the number of cache misses),
and its {\it span} (i.e., the length of the longest path in the computation DAG). 
The span is also the number of parallel steps in the computation assuming that unlimited number of processors 
are available and they all execute at the same rate. We discuss the binary fork-join model in more detail in Section~\ref{sec:prelim} and Appendix~\ref{sec:comp-model}.

\subsection{Our Results}

We present highly parallel data-oblivious multithreaded algorithms that are also cache-efficient for a suite of fundamental computation tasks,   
starting with sorting.
Our results show how to get privacy for free 
for a range of tasks in this important parallel computation model. 
Further, all of our algorithms are 
 cache-agnostic~\cite{frigo1999cache}\footnote{``Cache-agnostic'' is
also called cache-oblivious in the algorithms literature~\cite{frigo1999cache}. 
In order to avoid confusion with data obliviousness, 
we will reserve the term `oblivious' for
data-oblivious, and we will use {\it cache-agnostic} in place of cache-oblivious.},
i.e., the algorithm
need not know the cache's parameters including the cache-line (i.e., block) size 
and cache size. 
Besides devising new algorithms, our work also makes
a conceptual contribution by creating a bridge between  
two lines of work from the cryptography and 
algorithms literature, respectively. We now state our results. 

\paragraph{Sorting.}
To attain our results, the most important 
building block is data-oblivious sorting.
We present a randomized data-oblivious, cache-agnostic CREW binary fork-join
sorting algorithm, {\sc Butterfly-Sort},
\hide{
with optimal work and cache 
complexity and high parallelism that
matches the span of SPMS~\cite{CR17}, which has the current best span 
known for cache-efficient multithreaded algorithms.  
}
which matches the work, span and cache-oblivious caching bounds of SPMS~\cite{CR17},
the current best insecure algorithm. 
These bounds are given in Theorem~\ref{thm:sort}. As noted in~\cite{CR17}, these bounds are optimal for
 work and cache complexity, and are within an $O(\log\log n)$ factor of optimality for span 
 (even for the CREW PRAM).
The key ingredient in {\sc Butterfly-Sort} is  {\sc Butterfly-Random-Permute (B-RPermute)}, which 
randomly permutes the input array without leaking
the permutation, and guided by a butterfly network. The overall  {\sc Butterfly-Sort} simply runs {\sc B-RPermute} and then 
runs SPMS on the randomly
permuted array.

In comparison with known cache-agnostic and data-oblivious
sorting algorithms~\cite{cacheoblosort}, {\sc Butterfly-Sort}  improves the 
span by an (almost) exponential factor --- the prior work~\cite{cacheoblosort} requires 
$n^\epsilon$ parallel runtime  
 for some constant $\epsilon \in (0, 1)$ 
even without the binary fork-join constraint. 
As a by-product, our work gives a new oblivious sort result
in the standard external-memory model:
we actually give the first oblivious sort algorithm 
that is optimal in both total work and cache complexity --- such a result
was not known before even in the cache-aware setting.
For example, the prior works by Chan et al.~\cite{cacheoblosort}
and Goodrich~\cite{Goodrich-spaa11}
achieve optimal cache complexity, but 
incur extra $\log \log n$ factors in the total work; and moreover, Goodrich's result~\cite{Goodrich-spaa11} 
requires cache-awareness. Additionally, {\sc Butterfly-Sort} is a
conceptually simple algorithm and it is fairly straightforward to map it into a data oblivious algorithm matching
the current best sorting bound on other models such as BSP and PEM 
(Parallel External Memory)~\cite{AGNS08}.
We now state our 
sorting result.
\begin{theorem}[{\sc Butterfly-Sort}]\label{thm:sort}
Let $B$ denote the block size and $M$ denote the cache size.
Under the standard tall cache assumption
$M = \Omega (B^{1+\epsilon})$, and $M= \Omega (\log^{1+\epsilon} n)$ where
$\epsilon \in (0, 1)$ is an arbitrarily small constant, 
{\sc Butterfly-Sort}
 is a cache-agnostic CREW binary fork-join algorithm 
that obliviously sorts an array of size $n$
with an optimal cache complexity of $O((n/B) \cdot \log_M n)$,
optimal total work $O(n \log n)$, and $O(\log n \cdot \log\log n)$ span
which matches the best known non-oblivious algorithm in the same model.
\end{theorem}

\paragraph{Practical and EREW variants.}
We devise a conceptually simple  algorithm, {\sc Butterfly-Random-Sort (B-RSort)}, to sort
an input array {\it that has been randomly permuted}. 
This algorithm uses a collection of pivots, and 
is similar in structure to our {\sc B-RPermute} algorithm.
By replacing SPMS with {\sc B-RSort} to sort the randomly permuted output of {\sc B-RPermute},
we obtain a simple oblivious scheme for sorting, which we call {\sc Butterfly-Butterfly-Sort (BB-Sort)}. 
By varying the primitives used within {\sc BB-Sort}  we obtain
two useful versions: a potentially practical sorting algorithm which uses bitonic sort for small sub-problems, and an efficient EREW binary fork-join version
which retains AKS sorting for small subproblems. 

Our practical version uses bitonic sort as the main primitive, and for this we present an EREW  binary fork-join 
bitonic sort algorithm that improves on the na\"ive binary fork-join version by 
achieving span $O(\log^2 n\cdot \log\log n)$ and cache-agnostic caching
cost $O((n/B)  \cdot \log_M n \cdot \log (n/M))$ while retaining its $O(n \log^2 n)$ work. 

The use of AKS in our EREW and CREW algorithms may appear impractical.
However, it is to be observed that prior $O(n \log n)$-work oblivious algorithms 
for sorting --- the AKS network~\cite{aks}, Zigzag sort~\cite{zigzag}, and an $O(n \log n)$ version of oblivious bucket sort~\cite{bucketosort} ---
all use expanders within their 
construction. Sorting can be performed using an efficient oblivious priority queue that does not use 
expanders~\cite{pathoheap,opqkasper}; however this 
incurs $\omega(n \log n)$ work to achieve a negligible in $n$ failure probability.
Further, this method is not cache-efficient and is inherently sequential.
 It is to be noted that our practical version does not use AKS or expanders.
 
 Table~\ref{tab:sort} in Section~\ref{sec:practical} lists the bounds for our sorting algorithms.

\begin{table*}[t]
\centering
\caption{ {\bf Comparison with prior insecure algorithms.}
\label{tab:compare}
$\widetilde{O}(\cdot)$ {\rm hides a single $\log \log n$ factor.
LR = ``list ranking'', ET-Tree = ``Tree computations with Euler tour'', TC = ``Tree contraction'', 
CC = ``connected components'', MSF = ''minimum spanning forest''. 
For graph problems, $n$ is the number of vertices, and $m=\Omega(n)$ is the number of edges.
We compare with insecure cache-efficient CREW binary fork-join algorithms.
The prior bounds, except for tree contraction, are from~\cite{CR12}, and implicit in other work~\cite{lowdepthco,CRSB13}. The prior
result for sort is  SPMS sort~\cite{vijayasort,CR17}.
The prior bound for tree contraction (TC) is from~\cite{lowdepthco}. A `$^\dagger$' next to a result indicates that we improve the 
performance relative to  
the best known bound for the insecure case.
If cache-efficiency is not considered for insecure algorithms, the best CREW binary fork-join span for all tasks except Sort is
$O(\log^2 n)$; for  Sort it remains $\widetilde{O}(\log n)$.}
}
{\small
\begin{tabular}{c|ccc|ccc}
\toprule
\multirow{2}{*}{Task} & \multicolumn{3}{c|}{Our data-oblivious algorithm} & \multicolumn{3}{c}{Previous best insecure algorithm}\\
                      &  work & span & cache  &  work & span & cache\\
\midrule
Sort &  $O(n \log n)$  & $\widetilde{O}(\log n)$ & $O(\frac{n}{B} \log_M n)$ &
 $O(n \log n)$  & $\widetilde{O}(\log n)$ & $O(\frac{n}{B} \log_M n)$\\
LR
&  $O(n \log n)$  & $\widetilde{O}(\log^2 n)$ & $O(\frac{n}{B} \log_M n)$ &
 $O(n \log n)$  & $\widetilde{O}(\log^2 n)$ & $O(\frac{n}{B} \log_M n)$\\
ET-Tree &  $O(n \log n)$  & $\widetilde{O}(\log^2 n )$ & $O(\frac{n}{B} \log_M n)$ &
 $O(n \log n)$  & $\widetilde{O}(\log^2 n)$ & $O(\frac{n}{B} \log_M n)$\\
TC$^\dagger$
&  $O(n \log n)$  & $\widetilde{O}(\log^2 n)$ & $O(\frac{n}{B} \log_M n)$ &
 $O(n \log n)$  & $\widetilde{O}(\log^3 n)$ & $O(\frac{n}{B} \log_M n)$\\
CC$^\dagger$
&  $O(m \log^2 n)$  & $\widetilde{O}(\log^2 n)$ & $O(\frac{m}{B} \log_M n \log n)$ &
 $O(m \log^2 n)$ & $\widetilde{O}(\log^3 n)$ & $O(\frac{m}{B} \log_M n \log n)$\\
MSF$^\dagger$
&  $O(m \log^2 n)$  & $\widetilde{O}(\log^2 n)$ & $O(\frac{m}{B} \log_M n \log n)$ &
 $O(m \log^2 n)$ & $\widetilde{O}(\log^3 n)$ & $O(\frac{m}{B} \log_M n \log n)$\\
\bottomrule
\end{tabular}
}
\end{table*}

\begin{table*}[t]
\caption{\label{tab:oblcompare} {\bf Comparison with prior oblivious algorithms.}
{\rm The prior best is obtained by taking the best known oblivious PRAM algorithm 
and na\"ively fork and join $n$ threads at every PRAM step.
$\widetilde{O}$ hides a single $\log \log n$ or $\log\log s$ factor.
Aggr = aggregation, 
Prop = propagation, 
S-R = send-receive, 
PRAM = oblivious simulation of a $p$-processor,
$s$-space PRAM (cost of simulating a single step, assuming $s \geq p$), and 
$\bigstar = O(\log s \cdot ((p/B) \cdot \log_M p  + p \cdot \log_B s))$.
The best known algorithm for aggregation and propagation 
are due to \cite{graphsc,circuitopram}, 
send-receive is obtained by combining \cite{graphsc,circuitopram}, 
\cite{aks}, and \cite{cacheoblosort}, 
oblivious simulation of PRAM is due to or implied by \cite{opram,opramdepth}.}
}
{\small
\begin{tabular}{c|ccc|ccc}
\toprule
Obliv.& \multicolumn{3}{c|}{Our algorithm} & \multicolumn{3}{c}{Prior best}\\
     Alg.                 &  work & span & cache  &  work & span & cache\\
\midrule
Aggr & $O(n)$ & $O(\log n)$ & $O(n/B)$ & $O(n)$ & $O(\log^2 n)$ & $O(n/B)$ \\
\midrule
Prop & $O(n)$ & $O(\log n)$ & $O(n/B)$ & $O(n)$ & $O(\log^2 n)$ & $O(n/B)$ \\
\midrule
\multirow{2}{*}{Sort \& S-R} &           
 \multirow{2}{*}{$O(n \log n)$}  & \multirow{2}{*}{$\widetilde{O}(\log n)$} & 
\multirow{2}{*}{$O(\frac{n}{B} \log_M n)$} &
 $O(n \log n)$  & $O(\log^2 n)$ & $O(n \log n)$\\
& & & & ${O}(n \log n \log^2 \log n)$ & $O(n^\epsilon)$ & $O(\frac{n}{B} \log_M n)$\\
\midrule
\multirow{2}{*}{PRAM}& $O(s \log s)$& $\widetilde{O}(\log s)$ & $O(\frac{s}{B} \log_M s)$
& $O(s\log^2 s)$ & $O(\log^2 s)$ & $O(s \log s)$ \\ 
& $O(p \log^2 s)$ & $\widetilde{O}(\log s)$ &  
$\bigstar$ in caption
& $O(p \log^2 s)$ & $\widetilde{O}(\log^2 s)$ & $O(p \log^2 s)$\\
\bottomrule
\end{tabular}
}
\end{table*}

\paragraph{Data-oblivious simulation of PRAM in CREW binary fork-join.} 
Using our sorting algorithm,
we show how to  
compile any CRCW PRAM algorithm to a data-oblivious, cache-agnostic
binary-fork algorithm with non-trivial efficiency.
We present two results along these lines. The first is a compiler that works
efficiently for space-bounded PRAM programs, i.e.,  when the space
$s$ used is close to the number of processors $p$.  
We argue that this is an important special case because 
our space-bounded simulation of CRCW PRAM on oblivious binary fork-join ({\it space-bounded PRAM on OBFJ})
allows us to derive oblivious binary fork-join algorithms  
for several computational tasks
that are cornerstones of the parallel algorithms literature.
Our space-bounded PRAM on OBFJ result is given in Theorem~\ref{thm:pram-emul} in Section~\ref{sec:pram-sim}. As stated there,
each step of a $p$-processor CRCW PRAM can be emulated within the work, span, and cache-agnostic caching bounds for
sorting $O(p)$ elements.

Our second PRAM on OBFJ result works for the general case  
when the space $s$ consumed by the PRAM can be much greater than (e.g., a polynomial
function in) the number of processors $p$.
For this setting, we show a result that strictly generalizes  
the best known 
Oblivious Parallel RAM (OPRAM) construction~\cite{circuitopram,opramdepth}.
Specifically, we can compile any CRCW PRAM to a data oblivious, cache-agnostic, 
binary fork-join 
program where each parallel step in the original PRAM 
can be simulated with $p \log^2 s$ total work 
and $O(\log s \cdot \log \log s)$ span. The exact statement is in Theorem~\ref{thm:general-CRCW-on-obfj} in
Section~\ref{sec:opram}.
In terms of work and span, the bounds in this result
match the asymptotical performance of the best prior  
OPRAM result~\cite{circuitopram,opramdepth}, which is a PRAM on oblivious PRAM result, and
hence would only imply results for unbounded forking.\footnote{The concurrent work 
by Asharov et al.~\cite{optimuspram}
shows that assuming the existence of one-way functions, each parallel step
of a CRCW PRAM can be obliviously simulated 
in $O(\log s)$ work and $O(\log s)$
parallel time on a {\it CRCW} PRAM.
Their work is of a different nature because they consider computational security 
and moreover, their target PRAM allows concurrent writes.}.
Moreover, we also give explicit cache complexity bounds for our oblivious simulation
which was not considered in prior work~\cite{circuitopram,opramdepth}.

To obtain the above PRAM simulation 
results, we use some
building blocks that are core to the oblivious
algorithms literature, called {\it aggregation, propagation,}
and {\it send-receive}~\cite{graphsc,opram,circuitopram}. 
Table~\ref{tab:oblcompare}
shows the performance bounds of our algorithms for these building blocks
and general PRAM simulation vs. 
the best known {\it data-oblivious} results, where the latter  
is obtained by taking the best known oblivious PRAM algorithm and na\"ively 
forking $n$ threads in a binary-tree fashion
for every PRAM step. The table shows that we  
improve over the prior best results for all tasks.

\paragraph{Applications.}
There is a very large
collection of efficient PRAM algorithms in the literature for a variety of
important problems.
We use our PRAM simulation results 
as a generic way of translating some of these algorithms
into efficient data-oblivious algorithms in binary fork-join.
For other important computational problems we use our new sorting algorithm as a building block to directly obtain
efficient data-oblivious algorithms.

We list our results in Table~\ref{tab:compare}
and compare with results for insecure cache-efficient CREW binary fork-join algorithms. Our results for 
tree contraction (TC), connected components (CC), and minimum spanning forest (MSF)
are obtained through our space-bounded PRAM on OBFJ result, and in all three cases
our data-oblivious binary fork-join algorithms improve the commonly-cited span by a $\log n$ factor 
while matching the total work and cache complexity of the best known insecure algorithms.
For other computational tasks such as list ranking and rooted tree computations with Euler tour,
we design
data-oblivious algorithms that match the best work, span and cache-efficiency bounds known for insecure algorithms in the binary fork-join model: these give better bounds than what
we would achieve with the PRAM simulation.  

\elaine{TODO: add some technical highlight to say that one technique
an efficient bfj implementation of a butterfly network structure}

Throughout, we allow our algorithms
to have only an extremely small failure probability (either in correctness
or in security) that is $o(1/n^k)$, for any constant $k$. We refer to such a probability as being
{\it negligible} in $n$. Such strong bounds 
are typical for cryptographic and security applications, and are stronger than the w.h.p. bounds
(failure probability  $O(1/n^k)$, for any constant $k$) commonly used for standard (insecure) randomized algorithms.

\paragraph{Road-map.} The rest of the paper is organized as follows. The rest of this section  
has short backgrounds on
multithreaded, cache-efficient and data-oblivious  algorithms.
Section~\ref{sec:orba-sort} is on {\sc Butterfly-Random-Permute} and {\sc Butterfly-Sort}, our most efficient data oblivious
sorting algorithm. Section~\ref{sec:practical} is on 
{\sc Butterfly-Random-Sort} which gives our simple oblivious algorithm for sorting a random permutation, and {\sc Butterfly-Butterfly-Sort} 
which gives our 
practical and EREW sorting algorithms. Section~\ref{sec:pram-sim} has PRAM simulations, and
Section~\ref{sec:app-main} has applications. Many details are 
in the appendices.

\subsection{Binary Fork-Join Model}
\label{sec:prelim}

For parallelism 
 we consider  a multicore environment consisting of
$P$ cores (or processors), each with a private cache of size $M$.
The processors communicate through an arbitrarily large shared memory.
Data is organized in blocks (or `cache lines') of size $B$.

We will  express parallelism using a multithreaded model  through paired fork and join operations
 (See Chapter 27 in Cormen et al.~\cite{clr} for an introduction to this model).
The execution of a binary fork-join algorithm starts at a single processor. 
A thread can rely on a {\it fork} operation to
spawn two tasks; and the forked parallel tasks can be executed by
other processors, as determined by a scheduler. 
The two tasks
will later {\it join}, and the 
corresponding join serves as a
synchronization point: both of the spawned tasks must complete before the computation
can proceed beyond this join. 
The memory accesses are CREW: concurrent reads are allowed at a memory location but no concurrent writes at a location.

The {\it work} of a multithreaded algorithm ${\sf Alg}$
 is the total number of operations it executes; equivalently,
it is the sequential time when it is executed on one processor. The {\it span} of ${\sf Alg}$, often called
its  
critical path length or parallel time $T_{\infty}$, is the number of parallel steps in the algorithm.
The cache complexity of a binary fork-join algorithm is the total number of cache misses incurred across
all processors during the execution. 

For cache-efficiency, we  will use the  cache-oblivious model in Frigo et al.~\cite{frigo1999cache}.
 As noted earlier, in this paper we will use the term {\it cache-agnostic} in place of
cache-oblivious, and reserve the term `oblivious' for data-obliviousness.
We assume a cache of size $M$  partitioned into blocks (or cache-lines) of size $B$.
 We will use $Q_{sort}(n) = \Theta ((n/B) \cdot \log_M n))$ to denote the optimal caching bound for sorting $n$ elements,
 and for a cache-agnostic algorithm to achieve this bound we need $M= \Omega (B^2)$ (or $M=\Omega (B^{1+\delta})$
 for some given arbitrarily small constant $\delta>0$).
 Several sorting algorithms with optimal work (i.e., sequential time) and cache-agnostic caching cost are known, including two
in~\cite{frigo1999cache}.

We would like to design binary fork-join algorithms with small work, small sequential caching cost, 
preferably cache-agnostic, and small span. As explained in Appendix~\ref{sec:cache-eff-bfork},
this will lead to good parallelism and cache-efficiency in an execution under randomized work stealing~\cite{bfork04}.

Let $T_{\rm sort}(n)$ be the smallest span for a binary fork-join algorithm that sorts $n$ elements with a comparison-based sort.
It is readily seen that $T_{\rm sort}(n) = \Omega (\log n)$.
The current best binary fork-join algorithm for sorting is SPMS (Sample Partition Merge Sort)~\cite{CR17}. 
This algorithm has span $O(\log n \cdot \log \log n)$ with
optimal work $W_{\rm sort} = O(n \log n)$ and optimal cache complexity $Q_{\rm sort}(n)$. But the SPMS algorithm is not data-oblivious.
No non-trivial data oblivious parallel 
sorting algorithm with optimal cache complexity was known prior to the results we present in this paper.

\hide{
\paragraph{Additional discussion on the modeling.}
One can consider multi-way forking in place of requiring binary forking, but in this paper we choose to
stay with binary fork-joins since multithreaded algorithms with nested binary fork-joins have
good scheduling results under randomized work-stealing~\cite{bfork04,bfork01,CR13}.
One can also consider recent extensions made to this model such as using atomics (e.g., {\tt test-and-set})  to achieve synchronizations outside of
joins complementing forks as in the recent binary forking model~\cite{optimalbfork}, or using 
concurrent writes~\cite{optimalbfork, AC+20} to reduce the span. 
We will not use this extra power in our algorithms. So, for the rest of this paper, 
we will deal with
a CREW multithreaded model with binary fork-joins: memory accesses are CREW; 
the forks expose parallelism; synchronization occurs only at joins;
and any pair of fork-join computations are either disjoint or one is nested within the other. 
We  will refer to such a computation as a {\it binary fork-join} algorithm, or simply a multithreaded algorithm.
}

More background on caching and multithreaded computations is given in Appendix~\ref{sec:comp-model}.

\subsection{Data Oblivious Binary Fork-Join Algorithms}
\label{sec:do-short}

Unlike the terminology ``cache obliviousness'',
data obliviousness captures the privacy requirement of a program.
As mentioned in the introduction,  we consider a multicore secure processor like Intel's SGX (with hyperthreading),  
and all data is encrypted to the secure processor's secret key at rest 
or in transit. 
Therefore, the adversary, e.g., a compromised operating system or a malicious
system administrator with physical access to the machine, 
cannot observe the contents of memory.
However, the adversary may control the scheduler that
schedules threads onto cores.
Moreover, it can observe 
1) the  computation DAG 
that captures the pattern of forks and joins, and 
2) the memory addresses accessed by all threads of the binary fork-join program.
The above observations 
jointly form the ``access patterns'' of the binary fork-join program.

We adapt the standard definition of data obliviousness~\cite{oram00,oram10,circuitopram} 
to the binary fork-join setting.

\begin{definition}[Data oblivious binary fork-join algorithm]
We say that a 
binary fork-join algorithm ${\sf Alg}$ data-obliviously realizes (or obliviously realizes)
a possibly randomized functionality
$\mcal{F}$,  iff there exists a simulator ${\sf Sim}$ such that 
for any input ${\bf I}$, 
the following two distributions have negligible statistical distance:
1) the joint distribution of ${\sf Alg}$'s output
on input ${\bf I}$ and the access patterns; and 2) 
$(\mcal{F}({\bf I}), \Sim(|{\bf I}|))$.
\end{definition}
For example, in an oblivious random permutation,
the ideal functionality $\mcal{F}$ is the functionality which, upon
receiving an input array, outputs a random permutation of the input.
Note that the simulator $\Sim$ knows only the length of the input ${\bf I}$
but not the contents of ${\bf I}$, i.e., the access patterns are simulatable
without knowing the input.
For randomized functionalities (e.g., random permutation),
it is important that the {\it joint} distribution of the output and access patterns 
(often called the ``real world'')
is statistically close to the {\it joint} distribution of the functionality's output
and simulated access patterns (often called the ``ideal world'').
Observe that in the ideal world, the functionality $\mcal{F}$ uses independent
coins from the simulator $\Sim$. This implies that in the real world,
the access patterns 
should be (almost) independent
of the computation outcome (e.g., the output permutation).

We stress that our notion of data obliviousness continues to hold if the 
adversary can observe the exact address a processor accesses,
even if that data is in cache.
In this way, our security notion  
can rule out attacks that try to glean secret information 
through 
many types of cache-timing attacks.
Our notion secures even against computationally unbounded adversaries, often
referred to as {\it statistical security} in the cryptography literature.

One sometimes inefficient way to design binary-fork-join algorithms 
is to take a Concurrent-Read-Exclusive-Write (CREW) PRAM algorithm,  
and simply fork $n$ threads in a binary-tree fashion to simulate every step
of the PRAM. If the original PRAM has $T(n)$ parallel runtime and $W(n)$ work,
then the same program has $T(n) \cdot \log n$ span and $W(n)$ work in a binary-fork-join model. 
Moreover, if the algorithm  
obliviously realizes some functionality $\mcal{F}$ on a CREW PRAM, 
the same algorithm obliviously realizes $\mcal{F}$ in a binary-fork-join model too. 
More details on data obliviousness are in Appendix~\ref{sec:d-o}.

\section{{\sc Butterfly-Sort}}\label{sec:orba-sort}

To sort the input array,
our CREW binary fork-join algorithm {\sc Butterfly-Sort}  first 
applies an oblivious random permutation to the input elements,
and then apply any (insecure) {\it comparison-based} sorting algorithm
such as SPMS~\cite{CR17} to the permuted array. It is shown in~\cite{bucketosort}
that this will give an oblivious sorting algorithm.

Our random permutation algorithm, {\sc Butterfly-Random-Permute (B-RPermute)} 
builds on an algorithm  for this
problem given in Asharov et al.~\cite{bucketosort} 
 to obtain improved bounds as well as cache-agnostic
cache-efficiency. To attain an oblivious random permutation, a key step in~\cite{bucketosort}
is to randomly assign elements to bins of $Z = \omega(\log n)$ 
capacity but without leaking the bin choices --- we call
this primitive {\it oblivious random bin assignment (ORBA)}. 
For simplicity, we shall assume $Z = \Theta(\log^2 n)$.

In Section~\ref{sec:summarymetaorba} we describe our improvements to the ORBA algorithm 
in~\cite{bucketosort}. In Section~\ref{sec:summarysort} we describe {\sc Butterfly-Random-Permute}
and {\sc Butterfly-Sort}.
\hide{
We now give an overview of our algorithm for ORBA, and then our overall oblivious random 
permutation algorithm. We first present {\sc Meta-ORBA}, which does not consider cache-efficiency,
and then show that its implementation as {\sc Rec-ORBA} in Section~\ref{sec:summaryrecorba} gives
the desired performance bounds.
}

\subsection{Oblivious Random Bin Assignment: ORBA }
\label{sec:summarymetaorba}

We first give a brief overview of an
 ORBA algorithm by Asharov et al.~\cite{bucketosort}.
By using an algorithm for parallel compaction in~\cite{paracompact},
this ORBA algorithm in~\cite{bucketosort} can be made to
run with $O(n \log n)$ work and in $O(\log n \cdot \log\log n)$ parallel time on an EREW PRAM (though not on binary fork-join).
In this ORBA algorithm~\cite{bucketosort}, the $n$ input elements are divided
into $\bins = 2n/Z$ bins and each bin is then padded with $Z/2$ 
filler elements to a capacity of $Z$. 
We assume that 
$\bins$ is a power of 2.
Each real element in the input chooses a random label from the range $[0, \bins-1]$ 
which expresses the desired bin. The elements are then  
routed over a 2-way butterfly network of $\Theta(\log n)$ depth, and in each layer
of the butterfly network there are $\bins$ bins, each of capacity $Z$. 
In every layer $i$ of the network, input bins from the previous layer
are paired up in an appropriate manner, and the real elements  
contained in the two input bins 
are obliviously distributed to two output bins in the next layer, based on the $i$-th bit of
their random labels. For obliviousness, the output bins 
are also padded with filler elements to a capacity of $Z$.

Our `meta-algorithm' {\sc Meta-ORBA} saves an $O(\log\log n)$ factor in the  PRAM running time
while retaining the $O(n \log n)$ work.
For this we use a $\gamma = \Theta(\log n)$-way butterfly network, with $\gamma$ a power of 2, 
rather than $2$-way. Therefore,  
in each layer of the butterfly network, 
groups of $\Theta(\log n)$ input bins
from the previous layer 
are appropriately chosen, 
and the real elements in the $\log n$ input bins are obliviously distributed
to $\Theta(\log n)$ output bins in the next layer, 
based on the next unconsumed $\Theta(\log \log n)$ bits
in their random labels. 
Again, all bins are padded with filler elements to a capacity of $Z$
so the adversary cannot tell the bin's actual load.
To perform this $\Theta(\log n)$-way distribution obliviously, 
it suffices to invoke the AKS construction~\cite{aks} 
$O(1)$ number of times.

We now analyze the performance of this algorithm.

\boi
\item{\it  For $O(1)$ invocations of AKS on $\Theta(\log n)$ bins each of size $Z = \Theta(\log^2 n)$:}
 Total work is $O(\log^3 n \cdot \log \log n)$ and the parallel time on an EREW PRAM is $O(\log\log n)$.

\item  {\it For computing a single layer, 
with $\bins /\Theta(\log n)= 2n/\Theta(\log^3 n)$ subproblems:} 
 Total work for a single layer is  $O(n\cdot \log \log n)$ and EREW PRAM time remains $O(\log\log n)$.

\item The  $\Theta(\log n)$-way butterfly network has  
$\Theta(\log n/\log\log n)$ layers, 
hence {\sc Meta-ORBA}  performs $O(n\log n)$ total work, and runs in
$O(\log n)$ parallel time on an EREW PRAM.

\eoi

\paragraph{Overflow analysis.}
Our {\sc Meta-ORBA} algorithm has deterministic data access patterns
that are independent of the input.
However, if some bin overflows its capacity due to the random label assignment, 
the algorithm can lose real elements  
during the routing process.
Asharov et al.~\cite{bucketosort}
showed that for the special case $\gamma = 2$, 
the probability of overflow is upper bounded by $\exp(-\Omega(\log^2 n))$
assuming that the bin size $Z = \log^2 n$.
Suppose now that $\gamma = 2^r$, then  
the elements 
in the level-$i$ bins
in our ${\text{\sc Meta-ORBA}}$ algorithm correspond exactly to the elements 
in the level-$(r \cdot i)$ bins in  the ORBA algorithm in~\cite{bucketosort}.
Since the failure probability $\exp(-\Omega(\log^2 n))$
is negligibly small in $n$, we have that for $Z = \log^2 n$ and $\gamma = \Theta(\log n)$, 
our  ${\text{\sc Meta-ORBA}}$
algorithm obliviously realizes random bin assignment 
on an EREW PRAM with $O(n \log n)$ 
total work and $O(\log n)$ parallel runtime with negligible error.

We have now improved the parallel runtime by a $\log \log n$ factor
relative to Asharov et al.~\cite{bucketosort} (even when compared
against an improved version of their algorithm that uses 
new primitives such as parallel compaction~\cite{paracompact}) but this is on an EREW PRAM.
In the next section we translate this
  into  {\sc Rec-ORBA}, our cache-agnostic, binary fork-join implementation. 

\subsubsection{{\sc Rec-ORBA}}
\label{sec:summaryrecorba}

We will assume a tall cache ($M = \Omega (B^2)$) for {\sc Rec-ORBA}, our  recursive cache-agnostic, binary fork-join implementation of our {\sc Meta-ORBA} algorithm. We
will also assume $M= \Omega (n^{1 + \epsilon})$ 
for any given arbitrarily small constant $\epsilon >0$. In the following
we will use $\epsilon = 2$ for simplicity, i.e.,  $M = \Omega (\log^3 n)$.

In  {\sc Rec-ORBA},  we 
implement {\sc Meta-ORBA} by recursively solving $\sqrt \bins$ subproblems, 
each with $\sqrt \bins$
bins: in each subproblem, we shall obliviously distribute the  
real elements in the $\sqrt \bins$ 
input bins into $\sqrt \bins$
output bins, using the $(1/2) \log \bins$ most significant bits (MSBs) in the labels.
After this phase, we use a matrix transposition to bring 
the $\sqrt \bins$ 
bins with the same $(1/2) \cdot \log \bins$ MSBs together --- these
$\sqrt \bins$  bins now belong to the same subproblem for the next phase. 
where 
we recursively solve each subproblem defined above: for each subproblem, we 
distribute the $\sqrt \bins$  bins into $\sqrt \bins$  output bins
based on the $(1/2) \cdot \log \bins$ least significant bits
in the elements' random labels. 
The pseudocode is below.

\begin{mdframed}
\begin{center}
${\text{\sc Rec-ORBA}}^{\gamma, Z}({\bf I}, s)$
\end{center}
\noindent {\bf Input:}
the input contains an array ${\bf I}$
containing $n = \bins \cdot Z$ elements where $\bins$ is assumed to be a power of $2$.
Each element in ${\bf I}$ is either {\it real} or a {\it filler}.
The input array ${\bf I}$ can be viewed as $\bins$ bins each of size $Z$. 
It is promised that each bin has at most $Z/2$ real elements.
Every real element in ${\bf I}$ has a label that determines
which destination bin the element wants to go to.

\vspace{5pt}
\noindent {\bf Algorithm:}

If $\bins \leq \gamma$,  invoke an instance of oblivious bin placement
(Section~\ref{sec:oblivbldgblock})
to assign the input array ${\bf I}$ to a total of $\bins$ bins
each of capacity $Z$. 
Here, each element's destination bin 
is determined by the $s$-th to $(s+\log_2 \bins - 1)$-th bits of its label.
Return the resulting list of $\bins$ bins.

Else, proceed with the following recursive algorithm. 
\begin{itemize}[leftmargin=5mm,itemsep=1pt,topsep=1pt]
\item 
Let $\bins_1$ be $\sqrt{\bins}$ rounded up to the nearest power of 2.
Divide the input array ${\bf I}$ into 
$\bins_1$ partitions where each partition 
contains exactly $\bins_2 := \bins/\bins_1$ consecutive bins.
Note that if $\bins$ is a perfect square, $\bins_1 = \bins_2= \sqrt{\bins}$.
Henceforth let ${\bf I}^j$ denote the $j$-th partition where $j \in [\bins_1]$.
\item 
In parallel\footnote{The ``for'' loop is executed in parallel, 
and in the binary fork-join model, a $k$-way parallelism is achieved
by forking $k$ threads in a binary-tree fashion in $\log_2 k$ depth.
In the rest of the paper, we use the same convention when writing pseudo-code
for binary fork-join algorithms; 
we will use fork and join in the pseudocode for two-way parallelism.}: 
For each $j \in [\bins_1]$, let ${\sf Bin}_{1}^j, \ldots, {\sf Bin}_{\bins_2}^j 
\leftarrow {\text{\sc Rec-ORBA}}^{\gamma, Z}({\bf I}^j, s)$.
\item 
Let ${\sf Bins}$ be a ${\bins_1} \times {\bins_2}$ matrix
where the $j$-th row is the list of bins ${\sf Bin}_{1}^j, \ldots, {\sf Bin}_{\bins_2}^j$.
Note that each element in the matrix is a bin.
Now, perform a matrix transposition: ${\sf TBins} \leftarrow {\sf Bins}^T$.\\
Henceforth, 
let ${\sf TBins}[i]$ denote the $i$-th row of ${\sf TBins}$ consisting of ${\bins_2}$ bins, 
\item 
In parallel: For $i \in [{\bins_2}]$, 
let $\widetilde{\sf Bin}_{1}^i, \ldots, \widetilde{\sf Bin}_{{\bins_1}}^i  
\leftarrow {\text{\sc Rec-ORBA}}^{\gamma, Z}({\sf TBins}[i], s + \log_2 \bins_2 - 1)$
\item 
Return the concatenation 
$\widetilde{\sf Bin}_{1}^1, \ldots, \widetilde{\sf Bin}_{{\bins_1}}^1$, 
$\widetilde{\sf Bin}_{1}^2, \ldots, \widetilde{\sf Bin}_{{\bins_1}}^2$,  
$\ldots$, 
$\ldots$, 
$\widetilde{\sf Bin}_{1}^{{\bins_2}}$, $\ldots$, $\widetilde{\sf Bin}_{{\bins_1}}^{{\bins_2}}$.
\end{itemize}
\end{mdframed}

Since the matrix transposition for $Y$ bins 
each of capacity $Z = \Theta(\log^2 n)$ can be performed
with $O(Y\cdot Z/B)$ cache misses, and $O(\log (Y\cdot Z)) = O(\log n)$
span, we have the 
following recurrences that characterize the cost of {\sc Rec-ORBA}
where $Y$ denotes the current number of bins, and $Q(Y)$ 
and $T(Y)$ denote
the cache complexity and span, respectively,  to solve a subproblem 
containing $Y$ bins:
\[
\begin{array}{l}
Q(Y) = 2 \sqrt{Y} \cdot Q(\sqrt{Y}) + O((Y \cdot \log^2 n)/B)\\[5pt]
T(Y) = 2 \cdot T(\sqrt{Y}) + O(\log (Y \cdot \log^2n))
\end{array}
\]

The base conditions are as follows.
Since $M=\Omega (\log^3 n)$ each individual 
$\Theta(\log n)$-way distribution 
instance fits in cache and incurs $O((1/B) \log^3n)$ cache misses. 
Hence we have the base case $Q(Y) = O(Y \log^2 n/B)$ when $Y \log^2 n \leq M$.
For the span, each individual
$\Theta(\log n)$-way distribution instance works on 
$O(\log^3 n)$ elements and has $O( \log^2\log n)$ span under
binary forking, achieved by forking and joining $\log^3 n$ tasks at each level 
of the AKS sorting network. 

Therefore, we have that 
 for $\beta = 2n/Z$: 
$Q(\bins) = O((n/B) \log_M n)$
and $T(\bins) = O(\log n\cdot \log\log n)$. 
Finally, if we want $M = \Omega(\log^{1+\epsilon} n)$ rather
than $M = \Omega(\log^3 n)$, we can simply parametrize 
the bin size to be $\log^{1+\epsilon/2}n$
and let $\gamma = \log^{\epsilon/2}n$.
This gives rise to the following lemma.

\begin{lemma}\label{orba-recursive}
Algorithm {\sc Rec-ORBA} has cache-agnostic   caching complexity $O((n/B) \log_M n)$ provided the tall cache has
 $M = \Omega (\log^{1+\epsilon} n)$, for any given
positive constant $\epsilon$. The algorithm performs $O(n \log n)$ work and has span $O(\log n \cdot \log\log n)$.
\end{lemma}

\subsection{{\sc Butterfly-Random-Permute} and {\sc Butterfly-Sort}}
\label{sec:summarysort}
From ORBA's output ({\sc Meta-ORBA} or {\sc Rec-ORBA})  we obtain $\beta$  bins where
each bin contains real and filler elements. 
To obtain a random permutation of the input array,  it is shown in~\cite{cacheoblosort,bucketosort} 
that it suffices to do the following:
Assign a $(\log n \cdot \log \log n)$-bit random label to each element, 
and sort the elements 
within each bin based their labels using 
bitonic sort, where all filler elements
are treated as having the label $\infty$ and thus moved to the end of the bin.
Finally, remove the filler elements from all bins, and output the result. 

We will sort the elements in each bin by their labels with bitonic sort  (with a cost of $O(\log\log n)$ for each comparison due to the large values
 for the labels). 
 We have $\beta$ parallel sorting problems on $\Theta(\log^2 n)$ elements, where each element has an $O(\log n \cdot \log\log n)$-bit label.
Each sorting problem can be performed  with $O(\log^4\log  n)$ span
 and $O(\log^2 n \cdot \log^3 \log n)$ work in binary fork-join using bitonic sort by paying a $O(\log\log n)$ factor for the work and span  for each comparison due
to the $O(\log n \cdot \log\log n)$-bit labels. The overall cost for this step across all subproblems
is  $O(n\cdot \log^3 \log n)$ work and the span remains $O(\log^4\log  n)$. Since each bitonic sort subproblem fits in cache, the caching 
cost is simply the scan bound $O(n/B)$.
These bounds are dominated by our bounds for {\sc Rec-ORBA}, 
so this gives an algorithm to generate a random permutation with the same bounds as {\sc Rec-ORBA}.

Finally, once the elements have been permuted in random order, 
we can use any insecure comparison-based sorting algorithm 
to sort the permuted array.
We use SPMS sort ~\cite{CR17}, the best cache-agnostic sorting algorithm for the binary fork-join model
(CREW)
to obtain {\sc Butterfly-Sort}.
Thus we achieve the bounds in Theorem~\ref{thm:sort} in the introduction,  and we have a data-oblivious sorting algorithm that 
 matches the performance bounds for SPMS~\cite{CR17}, the current best insecure
algorithm in the cache-agnostic binary fork-join model.  
Asharov et al.~\cite{bucketosort} and Chan et al.~\cite{cacheoblosort}
proved that given an ORP, one can construct an oblivious sorting algorithm
by first running ORP to randomly permute the input array, and then
applying any {\it comparison-based} non-oblivious sorting algorithm to sort
the permuted array.
It is important that the non-oblivious sorting algorithm adopted
be {\it comparison-based}, otherwise the resulting algorithm
may leak information through its access patterns (e.g., a non-comparison-based
sorting algorithm can look at some bits in each element and place
it in a corresponding bin).
The proof of Asharov et al.~\cite{bucketosort} 
and Chan et al.~\cite{cacheoblosort}
extends 
to the binary fork-join model in a straightforward manner.

\small{
\begin{table*}[t]
\caption{\label{tab:sort} {\bf Permute \& Sort.}
$\widetilde{O}$ {\rm hides a single $\log \log n$ factor. All algorithms except {\sc B-RSort}  are data-oblivious and all have negligible error, except {\sc Bitonic-Sort}, which is deterministic. Optimal bounds for general sorting (after the first two rows) are displayed in a box.}
}
{\footnotesize
\begin{tabular}{l|l|ccc}
\toprule
Algorithm & Function & work & span & cache\\
\midrule
{\sc B-RPermute}  &  $\begin{array}{l} \text{Outputs random permutation} \\ 
\text{of input array} \end{array}$ & $O(n \log n)$ & $\widetilde{O}(\log n)$ & $O(\frac{n}{B} \log_M n)$   \\
\midrule
{\sc B-RSort}  &  
$\begin{array}{l} \text{Sorts a randomly} \\ \text{permuted input array}\end{array}$  
& $ O(n \log n)$ &  $\widetilde{O}(\log n)$ & $O(\frac{n}{B} \log_M n)$  \\
\midrule
\hline
{\sc Butterfly-Sort} & $\begin{array}{l} 
\text{Sorts input array} \\
\text{using {\sc B-RPermute} and SPMS} \end{array}$ &  \fbox{$O(n \log n)$} & $\widetilde{O}(\log n)$ & \fbox{$O(\frac{n}{B} \log_M n)$}   \\
\midrule
{\sc Fork-join Bitonic-Sort} & 
$\begin{array}{l} \text{Based on deterministic} \\ \text{sorting network}\end{array}$
& $O(n \log^2 n)$ & $\widetilde{O}(\log^2 n)$ & $O(\frac{n}{B} \log_M n \cdot \log \frac{n}{M})$ \\
\midrule
\multirow{2}{*}{{\sc BB-Sort} } &           
\multirow{2}{*} 
{$\begin{array}{l} 
\text{Sorts using {\sc B-RPermute}} \\
\text{and {\sc B-RSort}} 
\end{array}$}
& Practical: $\widetilde{O}(n \log n) $  & $\widetilde{O}(\log^2 n)$ & \fbox{$O(\frac{n}{B} \log_M n)$} \\
&  & EREW: \fbox{$O( n \log n)$} &  $\widetilde{O}(\log^2 n)$ & \fbox{$O(\frac{n}{B} \log_M n)$} \\
\bottomrule
\end{tabular}
}
\end{table*}

}

\section{Practical and EREW  Sorting} 
\label{sec:practical}

So far, our algorithm relies on AKS~\cite{aks} and the SPMS~\cite{CR17} algorithm
as blackbox primitives and thus is not suitable for practical implementation. 
We now describe a potentially practical variant 
that is self-contained and implementable, and gets rid of both AKS and SPMS.
To achieve this, we make two changes:

\boi
\item {\it Bitonic Sort.} Our efficient method uses bitonic sort within it, so we first we give an efficient
 cache-agnostic, binary fork-join implementation
of bitonic sort. A na\"ive binary fork-join implementation of bitonic sort would incur
$O((n/B)\cdot \log^2 n)$ cache misses and $O(\log^3 n)$ span (achieved by
forking and joining the tasks in each layer in the bitonic network). In our
binary fork-join bitonic sort,  the work remains $O(n \cdot \log^2 n)$ but  the span
reduces to $O( \log^2 n\cdot \log\log n)$ and the caching cost
reduces to $O((n/B) \cdot \log_M n \cdot \log (n/M) )$. (See Section~\ref{sec:bitonic-main}.)

\item {\it Butterfly Random Sort.} Second, we present a simple algorithm to sort a randomly permuted input.
This algorithm 
{\sc Butterfly-Random-Sort} (or {\sc  B-RSort}) 
{\it uses the same structure as {\sc B-RPermute }}
but uses a sorted set of pivots to
determine the binning of the elements.  Outside of the need to initially
sort $O(n/\log n)$ random pivot elements, this algorithm has the same
performance bounds as {\sc B-RPermute}. (See Section~\ref{sec:recsort}.)

\eoi

We obtain a simple sorting algorithm, which we call {\sc Butterfly-Butterfly-Sort} or {\sc BB-Sort} by 
running {\sc B-RPermute} followed by {\sc B-RSort} on the input: The first call outputs a random permutation of
the input and the second call sorts this sequence since {\sc B-RSort} correctly sorts a randomly permuted input.
We highlight two implementation of {\sc BB-Sort}:

{\bf Practical BB-Sort.}
By using our improved bitonic sort algorithm in place of AKS networks for small subproblems in {\sc BB-Sort}
(both within {\sc B-RPermute} and
{\sc B-RSort}), we
obtain a simple and practical data-oblivious sorting algorithm that has
optimal cache-agnostic cache complexity if $M = \Omega (\log^{2+\epsilon} n)$
 and incurs only an $O(\log \log n)$ blow-up in work
and an $O(\log n)$ blow-up in span. This algorithm has small constant
factors, with each use of bitonic sort contributing a constant factor of 1/2 to the
bounds for the number of comparisons made.

\vone
{\bf EREW BB-Sort.} For a more  efficient EREW binary fork-join sorting algorithm, we retain the AKS network in the butterfly computations
in both {\sc B-RPermute} and {\sc B-RSort}. We continue to sort the pivots at the
start of {\sc B-RSort} with bitonic sort. Thus the cost
of the overall algorithm is dominated by the cost to sort the $n/\log n$ pivots using bitonic sort,  and we obtain a simple
self-contained EREW binary fork-join sorting algorithm with optimal work and cache-complexity and $\tilde{O}(\log^2 n)$ span.

Theorem~\ref{thm:simple-sort} states the bounds on these two variants. In the rest of this section
we briefly discuss bitonic sort in Section~\ref{sec:bitonic-main} and then we present and analyze {\sc B-RSort}
in Section~\ref{sec:recsort}.

\begin{theorem}\label{thm:simple-sort} For $M= \Omega (\log^{2+\epsilon} n)$:\\
(i) Practical BB-Sort runs in $O(n \log n  \log \log n)$ work,  cache-agnostic caching cost $Q_{sort}(n)$,
and $O(\log^2n \cdot \log\log n)$ span.\\
(ii) EREW BB-Sort runs in $O(W_{sort}(n))$ work, cache-agnostic caching cost $O(Q_{sort}(n))$, and 
 $O(\log^2n \cdot \log\log n)$ span.\\
 Both algorithms are data oblivious with negligible error.
 \end{theorem}

Table~\ref{tab:sort} lists the bounds for our sorting and permuting algorithms.

\subsection{Binary Fork-Join Bitonic Sort}\label{sec:bitonic-main}

Bitonic sort has $\log n$ stages of {\sc Bitonic-Merge}.
We observe that {\sc Bitonic-Merge} is a butterfly network and hence
it has a binary fork-join algorithm with $O(n\log n)$ work,
$O(Q_{sort}(n))$ caching cost, and $O(\log n \log\log n)$ span.
Using this within the bitonic sort algorithm
we obtain the
following theorem. (Details are in Appendix~\ref{sec:bitonic-merge}.)

\begin{theorem}
There is a data oblivious cache-agnostic binary fork-join implementation of bitonic sort
that can sort ${n}$ elements in $O({n} \log^2 {n})$
work, $O(\log^2 {n} \cdot \log \log {n})$ span, and $O(({n}/B) \cdot \log_M {n} \cdot \log ({n}/M))$ cache complexity,  when $n > M \geq B^2$.
\label{thm:bitonic}
\end{theorem}

\subsection{Sorting a Random Permutation: {\sc Butterfly-Random-Sort (B-RSort)}}\label{sec:recsort}

As mentioned, 
{\sc Butterfly-Random-Sort (B-RSort)}
uses the same network structure as {\sc Rec-ORBA}.
During pre-processing, we pick and sort a random set of pivot elements.
Initially, elements in the input array ${\bf I}$ are partitioned
in bins containing $\Theta(\log^3 n)$ elements each.  
Then, the elements traverse
a $\gamma$-way butterfly network where $\gamma = \Theta(\log n)$ and is
chosen to be a power of $2$: at each step, we use  sorting
to distribute a collection of $\gamma$ bins at the previous level
into $\gamma$ bins at the next level.
Instead of determining the output bin at each level by the random label assigned to an element as in 
{\sc Rec-ORBA}, 
here each bin 
has a range determined by two pivots, and each element will be
placed in the bin whose range contains the element's value.
Also,  in {\sc Rec-ORBA} we needed to use filler elements
to hide the actual load of the intermediate bins.
{\sc B-RSort} does not need to be data-oblivious 
since its input is a random permutation of the original input, hence
we  {\it do not need filler elements} in {\sc B-RSort}.

\paragraph{Pivot selection.}
The pivots are chosen in a pre-processing phase, 
so that we can guarantee that every bin has $O(\log^3 n)$ elements with all but negligible
failure probability, as follows.
\hide{
we show below.
Roughly speaking, the pivots are chosen to be approximate $\Theta(n/\log^3 n)$-quantiles,
which divide the input elements into $\Theta(n/\log^3 n)$ subsets of approximately equal size.
To select the pivots, we do the following:
}
\begin{enumerate}
\item 
First, generate a sample $\Pi$ of size close to $n/\log n$ from
the input array ${\bf I}$ 
by sampling each element with probability $1/\log n$. We then sort $\Pi$ using 
bitonic sort.
\item
In the sorted version of $\Pi$, every element whose index
is a multiple of $\log^2 n$ 
is moved into a new sorted array ${\sf pivots}$. 
Pad the ${\sf pivots}$ array with an appropriate number of $\infty$ pivots
such that its length plus 1 would be a power of $2$.
\end{enumerate}

Using a Chernoff bound it is readily seen that the size of $\Pi$ is $(n/\log n)~ \pm~ n^{3/4}$ except with negligible probability.
We choose every $(\log^2 n)$-th element after sorting $\Pi$ to form our set of pivots. Hence  $r$,  the number of pivots,
is $\frac{n}{\log^3n} + o(n/\log^3 n)$ except with negligible probability.

By Theorem~\ref{thm:bitonic}, sorting $\Theta(n/\log n)$ 
samples incurs $O({n} \log {n})$
total work, $O(({n}/B) \log_M {n})$ cache complexity, and span
$O(\log^2 {n}\log \log {n})$  --- this step will dominate the span
of our overall algorithm.
The second step of the above algorithm can be performed using
prefix-sum, incurring total work $O(n)$, cache complexity
$O(n/B)$, and span $O(\log n)$.

\paragraph{Sort into bins.}
Suppose that $r - 1$ pivots were selected above; the pivots
define a way to partition the values being sorted
into $r$ roughly evenly loaded {\it regions}, 
where the $i$-th region includes
all values in the range $({\sf pivots}[i-1], {\sf pivots}[i]]$  for $i \in [r]$.
For convenience, 
we assume ${\sf pivots}[0] = -\infty$ and 
${\sf pivots}[r] = \infty$. We first describe algorithm
 {\sc Rec-SBA} (Recursive Sort Bin Assignment), which partially
sorts the randomly permuted input into a sequence of bins, and then we sort
within each bin to obtain the final sorted output. This is similar to applying
{\sc Rec-ORBA} followed by sorting within bins in {\sc B-RPermute}.

In algorithm {\sc Rec-SBA} the input array ${\bf I}$ will be viewed
as $r = \Theta(n/\log^3 n)$ bins   
each containing $n/r = \Theta(\log^3 n)$ elements.
The {\sc Rec-SBA} algorithm first partitions the input bins into 
$\sqrt{r}$ groups each containing $\sqrt{r}$ consecutive bins. Then,
it recursively   
sorts each group using appropriate pivots among the precomputed pivots.
At this point, it applies a matrix transposition on the resulting 
bins. 
After the matrix transposition all elements in each group of $\sqrt{r}$ consecutive bins
will have values in a range between a pair of pivots $\sqrt r$ apart in the sorted list of pivots,
and this group is in its final sorted position relative to the other groups.
Now, for a second time, the algorithm recursively sorts
each group of consecutive $\sqrt{r}$ bins, using the appropriate pivots 
and this will place each element in its final bin in sorted order. 
Our algorithm also guarantees
that the pivots are accessed in a cache-efficient manner.

We give a more formal description in $\text{\sc Rec-SBA}^\gamma({\bf I}, {\sf pivots})$.
Recall that $\gamma = \Theta(\log n)$ is the branching factor in the butterfly network.

\paragraph{Final touch.}
To output the fully sorted sequence, 
{\sc BB-Sort} simply applies bitonic sort to each bin in the output of {\sc Rec-SBA} and outputs the concatenated outcome.
The cost of this step is asymptotically absorbed 
by {\sc Rec-SBA} in all metrics.

\paragraph{Overflow analysis.}  
It remains to show that no bin will receive more than  
$C \log^3 n$ elements for some suitable constant $C > 1$ except with negligible probability.




Let us use the term {\sc Meta-SBA} to denote the non-recursive meta-algorithm for 
{\sc Rec-SBA}.
For overflow analysis, we can equivalently analyze {\sc Meta-SBA}.
 Consider a collection of $\gamma = \Theta(\log n)$ 
bins in the $j$-th subproblem in level $i-1$ of {\sc Meta-SBA} 
 whose contents are input to a bitonic sorter. Let us refer to this as group $(i-1,j)$. These elements will be 
 distributed into $\gamma$ bins in level $i$ using the $\gamma - 1$ 
pivots with label pair $(i-1,j)$.
 The elements in these $\gamma$ bins in group $(i-1,j)$ came from $\gamma^{i-1}$  
 bins  in the first level,
 each containing exactly $\log^2 n$ elements from input array ${\bf I}$. 
Let $U$ be the set of elements in
 these $\gamma^{i-1}$ bins in the first level, so size of $U$ is 
$u= \gamma^{i-1} \cdot n/r = \gamma^{i-1} \cdot \Theta(\log^3 n)$.
 
 Let us fix our attention on a bin $b$ in the $i$-th level
 into which some of the elements in group $(i-1, j)$ are distributed after the bitonic sort. 
 Consider the pair of pivots $p, q$ that are 
 used to determine the contents of bin $b$ (we allow $p = -\infty$ and $q = \infty$ to account
for the first and last segment).
This pair  of pivots is used to split the elements in the level $i-1$ group $(i-1,j)$ hence
 the pivots $p$ and $q$ are $r/\gamma^{i-1}$ apart in the sorted sequence  ${\sf pivots}[1..r]$.
We also know that the number of elements in ${\bf I}$ with ranks between any two successive pivots in ${\sf pivots}[1..r]$
is at most $\log^3 n +o(\log^3 n)$ except with negligible probability. 
Hence the number of elements in the input array ${\bf I}$ that have ranks 
 between the ranks of these two pivots is $k= (r/\gamma^{i-1}) \cdot \log^3 n$.
 Recalling that $r\leq 2n/\log^3 n$ 
(except with negligible probability),
we have  $k= (r/\gamma^{i-1}) \cdot \log^3 n \leq  2n/(\gamma^{i-1})$  
except with negligible probability.

 Let $Y_b$ be the number of elements in bin $b$. These are the elements in $U$ that have rank between the ranks of $p$ and $q$. The random variable
 $Y_b$ is binomially distributed on $u = |U|$ elements with success probability 
 at most 
$k/n = 2/\gamma^{i-1}$. 
Hence, $\E[Y_b] \leq u \cdot k/n = \Theta(\log^3 n)$ and using Chernoff bounds,
 $Y_b < \Theta(\log^3 n) + o(\log^3 n)$ 
with all but negligible probability.

\vspace{.1in}
\begin{mdframed}
\begin{center}
$\text{\sc Rec-SBA}^\gamma({\bf I}, {\sf pivots})$
\end{center}
\noindent {\bf Input:}
The input array ${\bf I}$
contains $\bins$ bins where $\bins$ is assumed to be a power of $2$.
The array ${\sf pivots}$ contains $\bins - 1$ number of pivots
that define $\bins$ number of regions, where the $i$-th region 
is $({\sf pivots}[i-1], {\sf pivots}[i]]$.
Each element in ${\bf I}$ will go to an output bin  
depending on which of the $\bins$ regions its value falls into.

\vspace{5pt}
\noindent{\bf Algorithm:}
If $\bins \leq \gamma$, 
use bitonic sort to assign the input array ${\bf I}$ to a total of $\gamma$ bins
based on ${\sf pivots}$. Return the resulting list of $\bins$ bins.

Else, proceed with the following recursive algorithm. 
\begin{enumerate}[leftmargin=5mm,itemsep=1pt,topsep=1pt]
\item 
Divide the input array ${\bf I}$ into 
$\sqrt{\bins}$ partitions where each partition contains exactly $\sqrt{\bins}$ 
consecutive bins\footnote{For simplicity, we assume
that $\bins$ is a perfect square in every recursive call.}
Henceforth let ${\bf I}^j$ denote the $j$-th partition.

\item 
\label{step:pivot1}
Let ${\sf pivots}'$ be constructed by taking every pivot in ${\sf pivots}$ whose
index is a multiple of $\sqrt{\bins}$.

In parallel: For each $j \in [\sqrt{\bins}]$, let ${\sf Bin}_{1}^j, \ldots, {\sf Bin}_{\sqrt{\bins}}^j 
\leftarrow \text{\sc Rec-SBA}^{\gamma}(I^j, {\sf pivots}')$.
\item 
\label{step:trans}
Let ${\sf Bins}$ be a $\sqrt{\bins} \times \sqrt{\bins}$ matrix
where the $j$-th row is the list of bins ${\sf Bin}_{1}^j, \ldots, {\sf Bin}_{\sqrt{\bins}}^j$.
Note that each element in the matrix is a bin.
Now, perform a matrix transposition: \\
${\sf TBins} \leftarrow {\sf Bins}^T$.\\
Henceforth, 
let ${\sf TBins}[i]$ denote the $i$-th row of ${\sf TBins}$ consisting of $\sqrt{\bins}$ bins, 
\item 
\label{step:pivot2}
In parallel: For $i \in [\sqrt{\bins}]$:

\quad let $\widetilde{\sf Bin}_{1}^i, \ldots, \widetilde{\sf Bin}_{\sqrt{\bins}}^i  
\leftarrow \text{\sc Rec-SBA}^{\gamma}({\sf TBins}[i], 
{\sf pivots}[(i-1) \sqrt{\bins} + 1 .. i \cdot \sqrt{\bins} - 1])$
\item 
Return the concatenation 
$\widetilde{\sf Bin}_{1}^1, \ldots, \widetilde{\sf Bin}_{\sqrt{\bins}}^1$, 
$\widetilde{\sf Bin}_{1}^2, \ldots, \widetilde{\sf Bin}_{\sqrt{\bins}}^2$,  
$\ldots$, 
$\ldots$, 
$\widetilde{\sf Bin}_{1}^{\sqrt{\bins}}$, $\ldots$, $\widetilde{\sf Bin}_{\sqrt{\bins}}^{\sqrt{\bins}}$.
\end{enumerate}
\end{mdframed}

\vspace{.1in}


 \paragraph{Performance analysis.}  
By the above overflow analysis we can assume that if any bin in  {\sc Rec-SBA} 
receives more than $C \log^3 n$ elements
for some suitable constant $C > 1$, the algorithm fails, since we have shown above that this
happens only with negligibly small in $n$ probability.

For the practical version, we replace 
 the AKS
with bitonic on problems of size $\Theta(\log^3 n)$, 
the total work becomes $O(n \log n \log \log n)$, and 
the span becomes 
$O(\log n \log \log n\log \log \log n)$. 

To see the cache complexity,  
observe that Line~\ref{step:pivot1} requires a sequential scan through the array ${\sf pivots}$
that is passed to the current recursive 
call, writing down $\bins - 1$ number of pivots, 
and making $\sqrt{\bins}$ copies of them to pass one to each of the $\sqrt{\bins}$ subproblems.
Now, Line~\ref{step:pivot2} divides ${\sf pivots}$
 into $\sqrt{\bins}$ equally sized partitions, removes the rightmost boundary 
point of each partition, 
passes each partition (now containing $\sqrt{\bins} -1$ pivots) to one subproblem.
Both Lines~\ref{step:pivot1} and \ref{step:pivot2} are upper bounded by the 
scan bound $Q_{\rm scan}(\beta \cdot \Theta(\log^3 n))$, and so is the matrix
transposition in Line~\ref{step:trans}. 
Therefore,  {\sc Rec-SBA} 
achieves optimal cache complexity similar to  {\sc Rec-ORBA}, but now 
assuming that $M = \Omega(\log^4 n)$. 


\paragraph{Our practical variant.}
Putting everything all together,
in our practical variant, 
we use {\sc RB-Permute} (instantiated with bitonic sort for each poly-logarithmically
sized problem), and similarly in {\sc Rec-SBA} and in the final sorting of the bins.
The entire algorithm --- outside of sorting the $O(n/\log n)$ pivots ---
incurs $O(n \log n \log \log n)$ total work, 
$O(\log n \log \log n\log \log \log n)$ span, and optimal cache complexity
$O((n/B) \log_M n)$.  When we add the cost of sorting the pivots, the
work and caching bounds remain the same but the span becomes $O(\log^ 2n \cdot \log\log n)$.
The constants hidden in the big-O are small:
for total work, the constant is 
roughly $18$; for cache complexity, the constant is roughly $6$;
and for span, the constant is roughly $1$.

\paragraph{Improving the requirement on $M$.}
We can improve the constraint on $M$
to $M= \Omega(\log^{2 + \epsilon} n)$
for an arbitrarily small $\epsilon \in (0, 1)$, if we make the following small modifications:
\begin{enumerate}[topsep=3pt,itemsep=1pt,leftmargin=5mm]
\item 
Choose every $(\log n)^{1+\epsilon}$-th element from the samples
as a pivot; and let $\gamma = \Theta(\log n)$.
Note that in this case, the total number of pivots is 
$r = \Theta(n/\log^{2 + \epsilon} n)$.

\item 
Earlier, we divided the input into 
$r$ bins each with $n/r$ elements; but now, we divide
the input array into $r \cdot \gamma$ bins, each filled with $n /(r \cdot \gamma) = 
\Theta(\log^{1 + \epsilon} n)$ elements.
Because there are $\gamma$ times more bins than pivots now, 
in the last level of the 
meta-algorithm, we will no longer have any pivots to consume -- but this does not matter
since we can simply sort each group of $\gamma$ bins in the last level.
\end{enumerate}

Finally, for the EREW version, we retain the AKS networks for the $polylog n$-size subproblems, so
the only step that causes an increase in performance bounds is the use of bitonic sort to sort $n/\log n$
pivots. Using our EREW binary fork-join algorithm for this step, we achieve the bounds stated in
Theorem~\ref{thm:simple-sort}.

\section{Oblivious, Binary Fork-Join Simulation of PRAMs}\label{sec:pram-sim}
We show that our new oblivious sorting algorithm
gives rise to oblivious simulation 
of CRCW PRAMs in the binary fork-join model with non-trivial efficiency guarantees.
In Section~\ref{sec:pram-app}  we give an oblivious simulation
of PRAM in the binary fork-join model 
that is only efficient if the space $s$ consumed by the original CRCW PRAM is small (e.g.,
roughly comparable to the number of cores $p$).
Then in Section~\ref{sec:opram}, we present 
another simulation that yields better efficiency when 
the original PRAM may consume large space.

\subsection{Oblivious, Binary Fork-Join Simulation of Space-Bounded PRAMs}\label{sec:pram-app}

We give an oblivious simulation
of PRAM in the  binary fork-join model 
that is only efficient if the space $s$ consumed by the original CRCW PRAM is small (e.g.,
roughly comparable to the number of cores $p$).
Our oblivious  simulation is based on a sequential cache-efficient  emulation of a PRAM step  given in~\cite{CG+95}, which in turn is based on a well-known
emulation of a $p$-processor CRCW PRAM on an EREW PRAM in $O(\log p)$ parallel time with $p$ processors (see, e.g.,~\cite{KR90}).
To ensure data-obliviousness, we will make use of the following well-known oblivious building block, 
{\it send-receive}:\footnote{The send-receive abstraction is often referred
to as oblivious routing in the data-oblivious 
algorithms literature~\cite{opram,circuitopram,opramdepth}. We avoid the name
``routing'' because of its other connotations in the algorithms literature.}. 

{\it send-receive}:
The input has a source array and a destination array.
The source array represents $n$ senders, each of whom holding
a key and a value; it is promised that all keys are distinct. 
The destination array  
represents $n'$ receivers each holding a key. 
Each receiver needs to
learn the value corresponding to the key it is requesting 
from one of the sources. If the key is not found, the receiver should receive $\bot$.
(Note that although each receiver wants only one value, a sender
can send its values to multiple receivers.) 

In Appendix~\ref{sec:oblivbldgblock}, we will show
how to accomplish 
oblivious send-receive within the sorting bound in the binary fork-join model. 
Using this primitive, the
PRAM simulation on binary fork-join works as follows.
We separate each PRAM step into a read step, followed by a local computation step (for which no simulation is needed), followed by a write step.
Suppose that in some step of the original CRCW PRAM, each of the $p$ 
processors has a request of the form $({\sf read}, {\addr}_i, i)$
or $({\sf write}, \addr_i, v_i, i)$ where $i \in [p]$, 
indicates that either it wants to read
from logical address ${\addr}_i$ or it wants to write $v_i$ to ${\addr}_i$.
\begin{enumerate}
\item 
For a read step, using oblivious send-receive, every request obtains a fetched value
from the memory array.
\item  
In a write step we need to perform the writes to the memory array.
To do this, we first perform 
a conflict resolution step in which we suppress the duplicate writes
to the same address in the incoming batch of $p$ requests. Moreover, if multiple
processors want to write to the same address, the  
one that is defined by the original CRCW PRAM's priority rule is preserved;
and all other writes to the same address are replaced with fillers.
This can be accomplished with $O(1)$ oblivious sorts.
Now, with oblivious send-receive, every address in the memory array can be updated with its new value. 
\end{enumerate}

Thus we have the following theorem.

\begin{theorem}[Space-bounded PRAM-on-OBFJ]\label{thm:pram-emul}
Let $M > \log^{1+\epsilon } s$ and $s\geq M \geq B^2$.
Any $p$-processor CRCW PRAM algorithm
that uses space at most $s$ can
be converted to a functionally equivalent, oblivious and cache-agnostic 
algorithm in the binary fork-join model, where each parallel step in the CRCW
can be simulated
with $O(W_{\rm sort}(p+s))$ work, $O(Q_{\rm sort}(p +s))$ cache complexity,
and $O(T_{\rm sort}(p+s))$ span.
\end{theorem}

Using the above theorem, any fast and efficient PRAM algorithm that uses linear space (and many
of them do) will give rise to an oblivious algorithm in the  
cache-agnostic binary fork-join model with good performance.
In Section~\ref{sec:app-main} we will use this theorem
to obtain new results for graph problems such as tree contraction,
graph connectivity, and minimum spanning forest ---
for all three problems, the performance bounds of our new data-oblivious 
algorithms improve on the previous best results 
even for algorithms that need not be data oblivious.

\subsection{Oblivious, Binary Fork-Join Simulation of Large-Space PRAMs} 
\label{sec:opram}
In this section, we describe a simulation strategy
that achieves better bounds when the PRAM's space can be large.
This is obtained by 
combining our 
Butterfly Sort 
with the prior work of 
Chan et al.~\cite{opramdepth}.
We assume familiarity with~~\cite{opramdepth} here;
Appendix~\ref{sec:chan-opram} 
has an
overview of this prior result.
\hide{
Specifically, we show that any $p$-processor CRCW PRAM 
consuming at most $s$ space can 
be converted to a functionally equivalent, data-oblivious and binary fork-join
program, where each parallel step in the CRCW
can be simulated 
with $O(p \log^2 s)$ total work, $O(\log s \cdot  \log \log s)$ span, 
and $O((p/B) \cdot \log_M p \log s + p \cdot \log s \cdot \log_B s)$ 
cache complexity.

Our result can be viewed as a generalization of the prior best 
Oblivious Parallel RAM (OPRAM) 
result~\cite{circuitopram,opramdepth}.
The prior result~\cite{circuitopram,opramdepth}
shows that any $p$-processor CRCW
PRAM can be simulated with an oblivious CREW PRAM 
where each CRCW PRAM step is simulated 
with $p \log^2 s$ total work and $O(\log s \log \log s)$ parallel runtime (without
the binary-forking requirement). 
Essentially our result matches the best known 
result, and we further show that the extra binary-forking requirement 
does not incur additional overhead during this simulation.
The prior work~\cite{circuitopram,opramdepth} also did not explicitly consider cache complexity.
}


To efficiently simulate any CRCW PRAM as a data-oblivious, binary fork-join program, 
we make the following modifications to
the algorithm in~\cite{opramdepth}.
First, we switch several core primitives they use to our cache-agnostic, binary
fork-join implementations:
we replace their oblivious sorting algorithm with our new 
sorting algorithm in the cache-agnostic, binary fork-join model; and
we replace their oblivious ``send-receive'', ``propagation'', and ``aggregation'' primitives with 
our new counterparts in the 
cache-agnostic, binary fork-join model (see Section~\ref{sec:app-main}).

Second, we make the following  modifications 
to improve the cache complexity:

\begin{itemize}[leftmargin=5mm,itemsep=1pt,topsep=3pt]
\item 
We store all  the binary-tree data structures
(called ORAM trees)
in Chan et al.~\cite{opramdepth} in 
an Emde Boas layout.
In this way, accessing a tree path of length $O(\log s)$ 
incurs only $O(\log_B s)$ cache misses.
\item 
The second modification is in the 
``simultaneous removal of visited elements'' 
step in the maintain phase of the algorithm in~\cite{opramdepth}.
In this subroutine, for each of $O(\log s)$ recursion levels: 
each of the $p$ processors populates 
a column of a $(\log s) \times p$ matrix, and then oblivious aggregation 
is performed on each row of the matrix.
To make this cache efficient, we can have each of the $p$ 
threads populate a row  
of a $p \times (\log s)$ matrix, and then perform matrix transposition.
Then we then apply oblivious aggregation to each row of the transposed
matrix.
\end{itemize}

Plugging in these modifications to Chan et al.~\cite{opramdepth}
we obtain the following theorem:

\begin{theorem}[PRAM-on-OBFJ: Simulation of CRCW PRAM on Oblivious, binary fork-join]\label{thm:general-CRCW-on-obfj}
Suppose that $M > \log^{1+\epsilon } s$, $s\geq M \geq B^2$, and $s \geq p$.
Any $p$-processor CRCW PRAM
using at most $s$ space can
be converted to a functionally equivalent, oblivious 
program in the cache-agnostic, binary fork-join model, where each parallel step in the CRCW
can be simulated
with $O(p \log^2 s)$ total work, $O(\log s \cdot \log \log s)$ span,
and $O(\log s \cdot ((p/B) \cdot \log_M p  + p \cdot \log_B s))$
cache complexity.
\end{theorem}
\begin{proof}
The total work is the same as Chan et al.~\cite{opramdepth}
since none of our modifications incur asymptotically more work.
\hide{
Specifically, all of our primitives, including
sorting, aggregation, propagation, and send-receive are optimal in total work.
Further, the modifications needed for cache complexity do not incur additional work.
}

Our span matches the PRAM depth of Chan et al.~\cite{opramdepth}, that is,
$O(\log s \log \log s)$, because all of our primitives, including
sorting, aggregation, propagation, and send-receive incur at most  
$O(\log s \log \log s)$ span.
The bottleneck in PRAM depth of Chan et al.~\cite{opramdepth}
comes not from the sorting/aggregation/propagation/send-receive, 
but from having to {\it sequentially} visit $O(\log s)$
recursion levels during the fetch phase, where for each recursion level,  
we need to look at a path in an ORAM tree of length $O(\log s)$, and find the
element requested along this path --- this operation takes $O(\log \log s)$ PRAM  
depth as well as $O(\log \log s)$
span in a binary fork-join model.
Finally, the new matrix transposition modification in the ``simultaneous removal'' step
does not additionally increase the span.

For cache complexity, 
there are two parts, the part that comes from 
$\log s$ number of oblivious sorts on $O(p)$ number of elements ---
this incurs 
$O((p/B) \cdot \log_M p\cdot \log s)$
cache misses in total.
The second part comes from having to access 
$p \cdot \log s$ ORAM tree paths in total where each path is of length $O(\log s)$.
If all the ORAM trees are stored in Emde Boas layout,
this part incurs 
$O(p \cdot \log s \cdot \log_B s)$ cache misses.
These two parts dominate the cache complexity, and the caching cost of all 
other operations is absorbed by the sum of these two costs.
\end{proof}

This result can be viewed as a generalization of the prior best 
Oblivious Parallel RAM (OPRAM) 
result~\cite{circuitopram,opramdepth}.
The prior result~\cite{circuitopram,opramdepth}
shows that any $p$-processor CRCW
PRAM can be simulated with an oblivious CREW PRAM 
where each CRCW PRAM step is simulated 
with $p \log^2 s$ total work and $O(\log s \log \log s)$ parallel runtime (without
the binary-forking requirement). 
Essentially our result matches the best known 
result, and we further show that the extra binary-forking requirement 
does not incur additional overhead during this simulation.
We additionally achieve cache-agnostic cache efficiency:
The prior work~\cite{circuitopram,opramdepth} also did not explicitly consider cache complexity.


\section{Applications}\label{sec:app-main}
\label{sec:apps}

We sketch our results for
 various applications of our new sorting algorithm,
assuming knowledge of results for the non-oblivious case. (Self-contained
descriptions of the algorithms are in the 
Appendix, Section~\ref{sec:app}.)
All of our  data-oblivious algorithms either match or outperform the best known bounds
for their insecure counterparts (in the cache-agnostic, binary fork-join model).
 Euler tour and tree contraction were  also considered 
in data-oblivious algorithms~\cite{oblgraph00,oblgraph01,blantongraph}
but the earlier works are inherently sequential.

{\bf Basic data-oblivious primitives.}   {\it Aggregration and propagation in a sorted array} were primitives
used in the simulation algorithms in Section~\ref{sec:opram}. These can be readily formulated as
segmented prefix sums and can be computed data obliviously within the scan bound ($O(\log n)$ span,
$O(n)$ work and $O(n/B)$ cache-agnostic cache misses for an $n$-array.) The third primitive 
used in~Section~\ref{sec:opram}, {\it send-receive},
 can be performed obliviously with a constant number of sorts and scans and hence 
achieves the sort bound. (See 
Appendix~\ref{sec:oblivbldgblock}).

{\bf List Ranking and Applications.}
To realize list ranking obliviously, 
we first apply {\sc Rec-OPerm} to randomly permute the input array and then apply 
a non-oblivious list ranking algorithm~\cite{CR12} to match the bounds for the insecure case:
$O(W_{sort}(n))$, $O(Q_{sort}(n))$, and $O(T_{sort}(n) \cdot \log n)$.

For Euler tour in a unrooted tree we use the insecure algorithm~\cite{KR90,CR12}  that reduces the problem to
list ranking. In this reduction there is an insecure step where each edge $(u,v)$ needs to  locate the successor of
edge $(v,u)$ in $v$'s circular adjacency list. This can be seen as an instance of send-receive and hence can be performed
obliviously within the sort bound, so we can compute the Euler tour obliviously with the same bounds as list ranking.
Once we have an Euler we can compute common tree functions (e.g., pre and post order numbering and depth of each
node) for any given root with list ranking.
Thus we obtain the bounds  stated in part $(i)$ of Theorem~\ref{thm:conn-msf-main} .

\hide{
  following theorem and we
match the state-of-the-art, non-oblivious, cache-agnostic, parallel 
list ranking algorithm~\cite{CR12}.

\begin{theorem}[List ranking]
Assume that $M \geq \log^{1+\epsilon} n$ and that $n \geq M \geq B^2$. 
Then, we can obliviously realize 
list ranking, Euler tour and many tree functions achieving span $O(\log^2 n \cdot \log \log n)$,
cache complexity $O((n/B) \log_M n)$, and
total work $O(n \log n)$.
Moreover, the algorithm is cache-agnostic.
\end{theorem}
}%

{\it Results Using PRAM Simulation.}
Using our PRAM simulation results in Section~\ref{sec:opram} we obtain oblivious algorithms with
 improved bounds for over what is generally claimed for the insecure algorithms.
Our data-oblivious algorithms are randomized  due to the use of our randomized sorting algorithm;
for the prior best insecure algorithms it suffices to use
the SPMS sorting algorithm~\cite{CR17} and therefore, they are deterministic (but not data oblivious).

\paragraph{Tree contraction.}
The EREW PRAM tree contraction algorithm of Kosaraju and Delcher~\cite{KD97} 
runs in  $\log n$ phases where the $i$-th phase, $i\geq 1$,  performs a constant number 
of parallel steps with  $O(n/c^{i-1})$
work on $O(n/c^{i-1})$ data items, for a constant $c>1$.
Further, the rate of decrease is fixed and data independent, and at
the end of each phase, 
every memory location knows whether it is still needed in future computation.
Since the memory used is geometrically decreasing in successive phases,
 a constant fraction of the memory locations  
will be no longer needed at the end of each phase.

We now apply our earlier PRAM simulation result in  Theorem~\ref{thm:pram-emul}  
in a slightly non-blackbox fashion: We  simulate each PRAM phase, 
always using the actual number of processors
needed in that step,  
which is the work for that step in the work-time formulation.
Additionally, we introduce the following modification:
at the end of each phase, we
use oblivious sort move the memory entries no longer needed to the end, 
effectively removing them from the future computation.
This gives the first result in  
in Theorem~\ref{thm:conn-msf-main} below.

\paragraph{Connected components and minimum spanning forest.}
These two problems can be solved on the PRAM on an $n$-node, $m$-edge graph
 in $T(n)=O(\log n)$ parallel time and with $O(m+n)$ space and number of processors~\cite{vishkinCC,PR02}.
Hence 
using Theorem~\ref{thm:pram-emul} for each of the $T(n)$ steps we achieve the bounds stated in part $(ii)$ of the
theorem below.  This improves over the generally cited bounds for the insecure case, which are obtained through explicit
algorithms and have span that is large (i.e. slower) by a $\log n$ factor.

\begin{theorem}\label{thm:conn-msf-main}
Suppose that $M > \log^{1+\epsilon } (m+n)$ and $m+n \geq M \geq B^2$.
Then, we have the following bounds for oblivious cache-agnostic, binary fork-join algorithms.

(i) List ranking, Euler tour and tree functions, and tree contraction in
$O(W_{\rm sort}(n))$ work, $O(Q_{\rm sort}(n))$ cache complexity,
and $O(\log n \cdot T_{\rm sort}(n))$ span, where $n$ is the size of the input. 

(ii)
Connected components and minimum spanning forest in a graph with $n$ nodes 
and $m$ edges in  $O(\log n \cdot W_{\rm sort}(m+n))$ work, 
$O(\log n \cdot Q_{\rm sort}(m+n))$ cache complexity, and
$O(\log n \cdot  T_{\rm sort}(m+n))$ span. 
\end{theorem}

\section*{\LARGE Appendices}

\appendix

\section{Computation Model}\label{sec:comp-model}
As mentioned earlier 
in Section~\ref{sec:intro}, binary fork-join is 
widely accepted 
in the algorithms literature as 
the model that 
best captures the 
performance of parallel algorithms on {\it multicore processors}.
We present some background on this model and how to measure
the efficiency of algorithms in this model.

\subsection{Multithreaded, Binary Fork-Join Computation Model}\label{sec:model}

We consider a multicore environment consisting of
$P$ cores (or processors), each with a private cache of size $M$.
The processors communicate through an arbitrarily large shared memory.
Data is organized in blocks (or `cache lines') of size $B$.

We will  express parallelism through paired fork and join operations (See Chapter 27 in Cormen et al.~\cite{clr} for an introduction to this model).
The execution of a binary fork-join algorithm starts at a single processor. 
A thread can rely on a {\it fork} operation to
spawn two tasks; and the forked parallel tasks can be executed by
other processors, as determined by a scheduler. 
The two tasks
will later {\it join}, and the 
corresponding join serves as a
synchronization point: both of the spawned tasks must complete before the computation
can proceed beyond this join. 
The memory accesses are CREW: concurrent reads are allowed at a memory location but no concurrent writes at a location.

The {\it work} of a multithreaded algorithm ${\sf Alg}$
 is the total number of operations it executes; equivalently,
it is the sequential time when it is executed on one processor. The {\it span} of ${\sf Alg}$, often called
its  
critical path length or parallel time $T_{\infty}$, is the number of 
parallel steps in the algorithm.

A widely used scheduler for multithreaded algorithms is
randomized work-stealing~\cite{bfork04}, which schedules the parallel tasks in a 
binary fork-join computation with
work $W$ and span $T_{\infty}$ on $P$ 
processors to run in $O((W/P) + T_{\infty})$ time with
polynomially small probability of failure in $P$, a result shown by 
Blumofe and Leiserson~\cite{bfork04}.

\paragraph{Additional discussion on the modeling.}
One can consider multi-way forking in place of requiring binary forking, but in this paper we choose to
stay with binary fork-joins since multithreaded algorithms with nested binary fork-joins have
good scheduling results under randomized work-stealing~\cite{bfork04,bfork01,CR13}.
One can also consider recent extensions made to this model such as using atomics (e.g., {\tt test-and-set})  to achieve synchronizations outside of
joins complementing forks as in the recent binary forking model~\cite{optimalbfork}, or using 
concurrent writes~\cite{optimalbfork, AC+20} to reduce the span. 
We will not use this extra power in our algorithms. So, in this paper
we deal with
a CREW multithreaded model with binary fork-joins: memory accesses are CREW; 
the forks expose parallelism; synchronization occurs only at joins;
and any pair of fork-join computations are either disjoint or one is nested within the other. 
We  will refer to such a computation as a {\it binary fork-join} algorithm, or simply a multithreaded algorithm.
Currently, for sorting, the best binary fork-join algorithm with optimal cache-efficiency is SPMS, and SPMS is
also the best cache-optimal sorting algorithm currently known even if atomics are allowed.

\subsection{Cache-efficient Computations}

 Memory in modern computers
 is typically organized in a hierarchy 
 with registers in the lowest level followed by several levels of 
 caches (L1, L2 and possibly L3), RAM, and disk.
 The access time and size of each level increases with its depth,
 and block transfers are used between adjacent levels to amortize
the access time cost.
 
\paragraph{External-memory algorithms.}
 The {\em two-level I/O model} in Agarwal and Vitter~\cite{AV88} is a simple abstraction
 of the memory hierarchy that consists of  a {\em cache} 
 (or {\em internal memory}) of size $M$, and an arbitrarily
 large {\em main memory} (or {\em external memory}) partitioned 
 into blocks of size $B$. An algorithm is said to have caused a 
 {\em cache-miss} (or {\em page fault}) if it references a block 
 that does not reside in the cache and must be fetched from the main memory.
 The {\em cache complexity} (or {\em I/O complexity}) of an algorithm is
 the number of block transfers or I/O operations it causes, which is
equivalent to the number of cache misses it incurs.
 Algorithms designed for this model often crucially
 depend on the knowledge of $M$ and $B$, and thus do 
 not adapt well when these parameters change --- such algorithms
are also called cache-aware algorithms.

\paragraph{Cache-agnostic algorithms.}
 The {\em cache-agnostic model} (originally called the cache-oblivious model),
proposed by 
Frigo et al.~\cite{frigo1999cache}, is an extension
 of the two-level I/O model which assumes that an optimal cache
replacement policy is used, and requires that the algorithm 
{\it be unaware of cache parameters $M$ and $B$}.
 A {\em cache-agnostic} algorithm is flexible and portable: 
if a cache-agnostic algorithm achieves optimal cache complexity,
it means that the number of I/Os are minimized  
in between any two adjacent levels in the memory hierarchy (e.g., between
CPU and memory, between memory and disk).
The assumption of an optimal cache replacement
 policy can be reasonably approximated by a
standard cache replacement method such as LRU.

\paragraph{Cache bounds for sort and scan, tall-cache assumption.}
Without the data obliviousness requirement and on a {\it sequential} machine,
the cache complexity for two basic computational tasks,
scan and sort, are known: the number
of cache-misses incurred to read $n$ contiguous data items from the main
memory is $Q_{\rm scan}(n) = \Theta(n /B)$ and the number of cache-misses incurred to
sort $n$ data items is $Q_{\rm sort}(n) = \Theta({n \over B}\log_{{M \over B}}
{{n \over B}})$~\cite{AV88} --- these bounds are both upper- and lower-bounds.
For most realistic values of $M$, $B$ and
$n$, e.g., under the tall cache-assumption which we also make, 
$Q_{\rm scan}(n) < Q_{\rm sort}(n) \leq n$.  
More specifically, the standard {\it tall cache} assumption
assumes that $M= \Omega (B^{1+\epsilon})$; and 
many works also additionally assume 
that $M = O(\log^{1+\epsilon} n)$ where $\epsilon \in (0, 1)$ is an arbitrarily
small constant.
The assumption $M= \Omega (B^{1+\epsilon})$ is also necessary
for obtaining the optimal sort bound.


\subsection{Cache-Efficiency in Binary Fork-Join}
\label{sec:cache-eff-bfork}
The cache complexity of a binary fork-join algorithm is the total number of cache misses incurred across
all processors during the execution. 
Let ${\sf Alg}$ be a binary fork-join algorithm with work $W$  and span $T_{\infty}$,
and let $Q$ be its cache complexity on a {\it sequential} machine, 
i.e., ${\sf Alg}$'s cache complexity when executed on one processor. 
If  algorithm ${\sf Alg}$  is executed on $P$ processors using randomized work stealing, then it is shown by Acar et al.~\cite{bfork01}
that the number of cache misses in this parallel execution is 
is $O(Q + (M/B)\cdot P \cdot T_{\infty})$ with polynomially low probability of failure.
Since we typically operate with an input size $n$ that is greater than the total available cache size, i.e., $n \geq P\cdot M$,
this bound will be close to or dominated by the caching cost $Q$ if the span is small.

The above scheduling result~\cite{bfork01} suggests the  following paradigm 
for designing efficient, parallel algorithms in the binary fork-join model:
we would like to design fork-join algorithms with small {\it work}, 
small {\it sequential cache complexity}, 
preferably {\it cache-agnostic}, and small {\it span}.

\paragraph{Cache-efficient sorting in the binary fork-join model.}
Let $T_{\rm sort}(n)$ be the smallest span for a binary 
fork-join algorithm that sorts $n$ elements with a comparison-based sort.
It is readily seen that $T_{\rm sort}(n) = \Omega (\log n)$.
The current best binary fork-join algorithm for sorting is SPMS (Sample Partition Merge Sort)~\cite{CR17}. 
This algorithm has span $O(\log n \cdot \log \log n)$ with
optimal work $W_{\rm sort} = O(n \log n)$ and optimal cache complexity $Q_{\rm sort}(n)$. But the SPMS algorithm is not data-oblivious.
No non-trivial data oblivious parallel 
sorting algorithm with optimal cache complexity was known prior to the results we present in this paper.

{\bf Additional Discussion.}
When designing cache-agnostic algorithms with parallelism, one would need to address the issue of
false-sharing, since without the knowledge of the block-size $B$, 
it is inevitable that a block that is written into could be 
shared across two or more processors. This would lead to the potential for false sharing. The design of algorithms that mitigate the cost of false
sharing is addressed in the resource-oblivious framework in the work of 
Cole and Ramachandran~\cite{CR12,CR17}, and the SPMS algorithm in~\cite{CR17} is designed to have
low worst-case false sharing cost. For the most part, the algorithms we develop here are in the 
Hierarchical Balanced Parallel (HBP)
formulation developed by Cole and Ramachandran~\cite{CR12,CR17}, and hence incur 
an overhead of no more than
$O(B)$ cache miss cost due to false sharing for each parallel task scheduled in the computation.
The one exception to HBP is the use of certain small subproblems, all of size smaller than $M$.

\section{Definitions: Data-Oblivious Binary Fork-Join Algorithms}\label{sec:d-o}
We define the notion of a data-oblivious binary fork-join algorithm. 
Unlike 
``cache obliviousness'', 
data obliviousness captures the privacy requirement of a program.
In our paper, we will consider {\it statistical} data-obliviousness, i.e.,
except with negligibly small probability, 
no information should be leaked through the program's ``access patterns'' (to be defined
later) 
even in the presence of a {\it computationally unbounded} adversary.

To understand the requirements of data obliviousness, 
it helps to think about a multicore machine with trusted hardware   
support (e.g., Intel SGX + hyperthreading). 
On such a secure, multicore architecture, 
any data that is written to memory is encrypted under a secret key  
known only to the secure processors. Therefore, unless the adversary can break  
into the secure processor's hardware sandbox, it cannot observe
any memory contents during the execution.
Such an adversary can observe
\begin{enumerate}[leftmargin=5mm,itemsep=0pt]
\item 
the computation DAG that captures the pattern of forks and joins 
during the course of the execution;
\item the sequence of memory {\it addresses}
accessed during 
every CPU step of every {\it thread} (but not the data contents), 
and whether each access is a read or write\footnote{As mentioned, in the binary fork-join model, the scheduler
decides which thread to map to which processor. 
We may assume that the adversary has full control of the scheduler.};
\item 
even when the processor is accessing words in the
cache, the adversary can observe which exact address 
within which block is being requested --- note that allowing
the adversary to observe this allows our data-oblivious algorithm
to secure against the well-known cache-timing attack (see also Remark~\ref{rmk:obldefn}). 
\end{enumerate}

The above observations jointly form the adversary's {\it view}
in the execution --- in our paper, we also call the adversary's view
the ``access patterns'' of the execution.

We now extend the data-oblivious framework developed in earlier work for sequential and PRAM parallel environments~\cite{oram00,oram10,optorama,asiacrypt11,circuitopram} to
binary fork-join algorithms.
Henceforth, we use the notation
$${\bf O}, {\sf view} \leftarrow  {{\sf Alg}({\bf I})}$$
to denote a possibly randomized execution of the algorithm ${\sf Alg}$
on the input ${\bf I}$.
The randomness of the execution can come from the algorithm's internal random coins.
The execution produces 
an output array ${\bf O}$, and the adversary $\algA$'s view during
the execution is denoted ${\sf view}$.

\begin{definition}[Data-oblivious algorithms in the binary-fork-join model]
We say that an algorithm ${\sf Alg}$ 
$\delta$-data-obliviously realizes 
 a possibly randomized functionality $\mcal{F}$, 
iff there exists a simulator $\Sim$ such that 
for any input ${\bf I}$ of length $n$, 
\[
({\bf O}, {\sf view}) \overset{\delta}{\equiv} (\mcal{F}({\bf I}), \Sim(n))
\]
where 
$({\bf O}, {\sf view})$ denotes the joint distribution
of the output and the adversary's view upon a random 
invocation of ${{\sf Alg}({\bf I})}$,
and 
$\overset{\delta}{\equiv}$ means that the left-hand side
and the right-hand side must have  
statistical distance at most $\delta$.
Note that $\delta$ can be a function 
of $n$, i.e., the length of the input. 
\label{defn:obl}
\end{definition}

\paragraph{Typical choice of $\delta$.}
Typically, we want the failure probability $\delta$ to be negligibly
small in 
the input length\footnote{Sometimes, we may want that the failure probability
be negligibly small in some security parameter $\kappa \in \N$ --- however, in this paper, 
we care about the asymptotical behavior of our algorithms
when the problem size $n$ is large.
Therefore, without loss of generality, we may assume that $n \geq \kappa$ and 
in this case, negligibly small in $n$ also means negligibly small in $\kappa$.
} 
denoted $n$.
A function $\delta(\cdot)$ is said to be a negligible 
function, iff it drops faster than any inverse-polynomial function.
More formally, we say that $\delta(\cdot)$ is a negligible function,
iff for any constant $c \geq 1$,  
there exists a sufficiently large $n_0 \in \N$ such that for all $n \geq n_0$,
$\delta(n) \leq 1/n^c$.

Whenever we say that ${\sf Alg}$ {\it data-obliviously 
realizes  (or obliviously realizes)}
a functionality $\mcal{F}$ omitting the failure probability $\delta$, we mean
that there exists a negligible function $\negl(\cdot)$, such 
that ${\sf Alg}$ $\negl(n)$-data-obliviously realizes $\mcal{F}$.

\begin{remark}[Real-world interpretation of the obliviousness definition]
On multicore architectures, 
the processors execute at asynchronous speeds due to various factors 
including cache misses, clock speed, etc. 
Since our model allows the adversary to observe even accesses 
within the processors' caches, our notion of data-obliviousness
precludes the adversary from gaining secret information through  
any cache-miss-induced timing channels.
However, just like in all prior
works on data-oblivious algorithms and Oblivious RAM~\cite{oram00,oram10,asiacrypt11,cacheoblosort}, 
data-obliviousness does not provide full protection against 
timing channel leakage caused by  
the processor's branch predictors, speculative execution, 
and other pipelining behavior that are data dependent
(although 
it has been shown that 
the lack of data-obliviousness
can sometimes exacerbate attacks that exploit 
timing channel differences caused by the processors' data-dependent pipelining behavior).

Note that the adversary, e.g., 
a malicious OS, also controls the scheduler that maps threads to available processors. 
However, since processors like SGX encrypt any data in memory, 
the adversary is unable to base scheduling decisions on data contents.
Indeed, the standard parallel
algorithms literature considers schedulers 
that are data-independent too~\cite{bfork04,bfork05,vijayasort,CR17,lowdepthco}.
\label{rmk:obldefn}
\end{remark}

\begin{remark}[Alternative definition for deterministic functionalities]
If the functionality $\mcal{F}$ is deterministic
(e.g., the sorting functionality), we can alternatively define
$\delta$-data-obliviousness with a slightly more straightforward definition.
We may say that an algorithm ${\sf Alg}$ $\delta$-data-obliviously 
realizes $\mcal{F}$, 
iff there exists a simulator $\Sim$,
such that 
for any possibly unbounded adversary $\algA$,
for any input ${\bf I}$, 
the following hold:
\begin{enumerate}
\item 
the algorithm ${\sf Alg}$
correctly computes the function $\mcal{F}$ with probability $1-\delta$;
and 
\item
the adversary's view in the execution ${\sf Alg}^\algA({\bf I})$
has statistical distance at most $\delta(n)$ away from the  
output of the simulator $\Sim(n)$.
Notice that since the simulator \Sim knows only the length of the input ${\bf I}$  
but not the input itself, 
this definition implies that the adversary can essentially simulate
its view in the execution  
on its own without knowing anything about the input ${\bf I}$ 
--- in other words,  nothing is leaked about the input ${\bf I}$
except for a small $\delta$ failure probability.
\end{enumerate}

Due to the definition and property of statistical distance, we immediately have the following.
If an algorithm ${\sf Alg}$ $\delta$-data-obliviously realizes
a deterministic functionality $\mcal{F}$ by the above alternative definition,
then it $2\delta$-data-obliviously realizes 
$\mcal{F}$ by Definition~\ref{defn:obl}.
Furthermore, if ${\sf Alg}$  
 $\delta$-data-obliviously realizes
$\mcal{F}$ by Definition~\ref{defn:obl}, 
it must 
 $\delta$-data-obliviously realizes
$\mcal{F}$ by the above alternative definition.
\end{remark}

\paragraph{Relationship to oblivious PRAM algorithms.}
One sometimes inefficient way to design binary-fork-join algorithms 
is to take a Concurrent-Read-Exclusive-Write (CREW) PRAM algorithm,  
and simply fork $n$ threads in a binary-tree fashion to simulate every step
of the PRAM. If the original PRAM has $T(n)$ parallel runtime,  
then the same program has $T(n) \cdot \log n$ 
span in a binary-fork-join model.
Moreover, if the algorithm  
$\delta$-obliviously realizes some functionality $\mcal{F}$ on a CREW PRAM, 
the same algorithm  
$\delta$-obliviously realizes $\mcal{F}$ in a binary-fork-join model too. 
In the above, we say that an algorithm 
$\delta$-obliviously realizes the functionality $\mcal{F}$
on a CREW PRAM iff Definition~\ref{defn:obl} holds, but 
the adversary's view now consists 
of the memory addresses 
accessed by all processors in all steps of the execution.

\begin{fact}
Suppose that an algorithm ${\sf Alg}$ 
$\delta$-obliviously realizes some functionality $\mcal{F}$
on a CREW PRAM
and it completes with total work $W(n)$ and parallel runtime $T(n)$.
Then, the same algorithm ${\sf Alg}$ (using  
a binary-tree to fork $n$ threads for every PRAM step), 
$\delta$-obliviously realizes $\mcal{F}$
in a binary-fork-join model, 
with total work $W(n)$ but span $T(n) \cdot \log n$.
\label{fct:opramtoobfork}
\end{fact}

\section{Additional Building Blocks}
\label{sec:oblivbldgblock}
To obtain oblivious simulation of 
CRCW PRAMs in the binary fork-join model, 
we will make use of three important primitives --- just like oblivious sorting, these
primitives have been core to the data-oblivious algorithms  
literature and in constructing Oblivious Parallel RAM 
schemes~\cite{circuitopram,opramdepth,graphsc,opram}:
\begin{itemize}[leftmargin=5mm,itemsep=1pt,topsep=3pt]
\item  {\it Aggregation in a sorted array.}
Given an array in which  
each element belongs to a group, and the array is sorted such that 
all elements belong to the same group appear consecutively, 
let each element learn the 
``sum'' of all elements belonging to its group, and appearing to its right.
In a more general formulation, ``sum'' can be replaced with any commutative
and associative aggregation function $f$.

\item  {\it Propagation in a sorted array.}
Given an array in which  
each element belongs to a group, with the array sorted such that 
all elements belong to the same group appear consecutively, we call
the leftmost element of each group the group's {\it representative}.
In the outcome, every element should output the value 
held by its group's representative, i.e., the representative propagated its value
to everyone in the same group.
\ignore{
let each element learn the 
``sum'' of all elements belonging to its group, and appearing to its right.
In a more general formulation, ``sum'' can be replaced with any commutative
and associative aggregation function $f$.
}

\item  {\it Send-receive\footnote{The send-receive abstraction is often referred
to as oblivious routing in the data-oblivious 
algorithms literature~\cite{opram,circuitopram,opramdepth}. We avoid the name
``routing'' because of its other connotations in the algorithms literature.}.}
In the input, there is a source array and a destination array.
The source array represents $n$ senders, each of whom holds
a key and a value; it is promised that all keys are distinct. 
The destination array  
represents $n'$ receivers each holding a key. 
Now, have each receiver
learn the value corresponding to the key it is requesting 
from one of the sources. If the key is not found, the receiver should receive $\bot$.
Note that although each receiver wants only one value, a sender
can send its values to multiple receivers. 
\item
{\it Oblivious bin placement.}
Bin placement 
realizes the following deterministic functionality~\cite{circuitopram}.
Let $\bins$ denote the number of bins, and let $Z$ denote the target bin capacity.
We are given an input array
denoted ${\bf I}$, where
each element is either a {\it filler} denoted $\bot$ or a {\it real} element
that is tagged with a bin identifier ${\group} \in [\bins]$ denoting which bin
it wants to go to.
It is promised that every bin will receive at most $Z$ elements.
Now, move each real element in ${\bf I}$ to
its desired bin.
  If any bin is not full after the real elements have been placed, pad it with filler elements {\it at the end}
to a capacity of $Z$. Output the concatenation of the $\bins$ bins as a single array.
An algorithm that obliviously realizes bin placement 
is called oblivious bin placement.
\end{itemize}

Prior works on data oblivious algorithms~\cite{graphsc,opram,circuitopram} 
have suggested oblivious PRAM constructions for the above primitives; but na\"ively
converting them to the binary fork-join model would incur a $\log n$
factor blowup in span.
To solve the aggregation and propagation, we 
can in fact use the segmented prefix (or suffix) sum primitive~\cite{jajabook}
which is well-known in the parallel algorithms literature, and can
be implemented as a prefix sum computation
incurring $O(n)$
total work, $O(n/B)$ cache complexity, and $O(\log n)$ span~\cite{CR12,CR17}.
Moreover, the algorithm naturally has fixed, data-independent access patterns.

To solve send-receive efficiently in the binary-fork join model, we can 
use the high-level construction 
by Chan et al.~\cite{circuitopram} who showed how to realize send-receive with $O(1)$
oblivious sorts and one invocation of oblivious propagation.
If we replace the oblivious sort and propagation primitives with efficient, binary fork-join  
implementations,  oblivious send-receive achieves the sorting bound.

Chan and Shi~\cite{circuitopram}
showed that oblivious bin placement can be realized with $O(1)$ oblivious sorts
and $O(1)$ oblivious propagation.

\section{Bitonic Sort in Binary Fork Join}
\label{sec:bitonic}

Our practical variant uses bitonic sort~\cite{batcher68}
rather than AKS in various places to solve smaller instances of the sorting problem.
We show a  data-oblivious and cache-agnostic binary fork-join implementation of bitonic sort
that achieves the following the bounds in Theorem~\ref{thm:bitonic}: 
 a data oblivious cache-agnostic binary fork-join implementation 
that can sort ${n}$ elements in $O({n} \log^2 {n})$
work, $O(\log^2 {n} \cdot \log \log {n})$ span, and $O(({n}/B) \cdot \log_M {n} \cdot \log ({n}/M))$ cache complexity,  when $n > M \geq B^2$.

For bitonic sort, the constants in the big-O notation are very small, roughly $1/2$ when counting comparisons for 
both the total work and span. 
A  cache-agnostic data-oblivious implementation of bitonic sort is given in~\cite{cacheoblosort} but 
our implementation
is asymptotically better both in cache complexity and span. In particular, our span
improves on the $O(\log^3 {n})$ bound that is obtained by converting each PRAM step
in the $O(\log^2 {n})$ time EREW PRAM  bitonic sort algorithm into a fork-join computation.

\begin{figure*}[t]
\begin{center}
\includegraphics[width=0.8\textwidth]{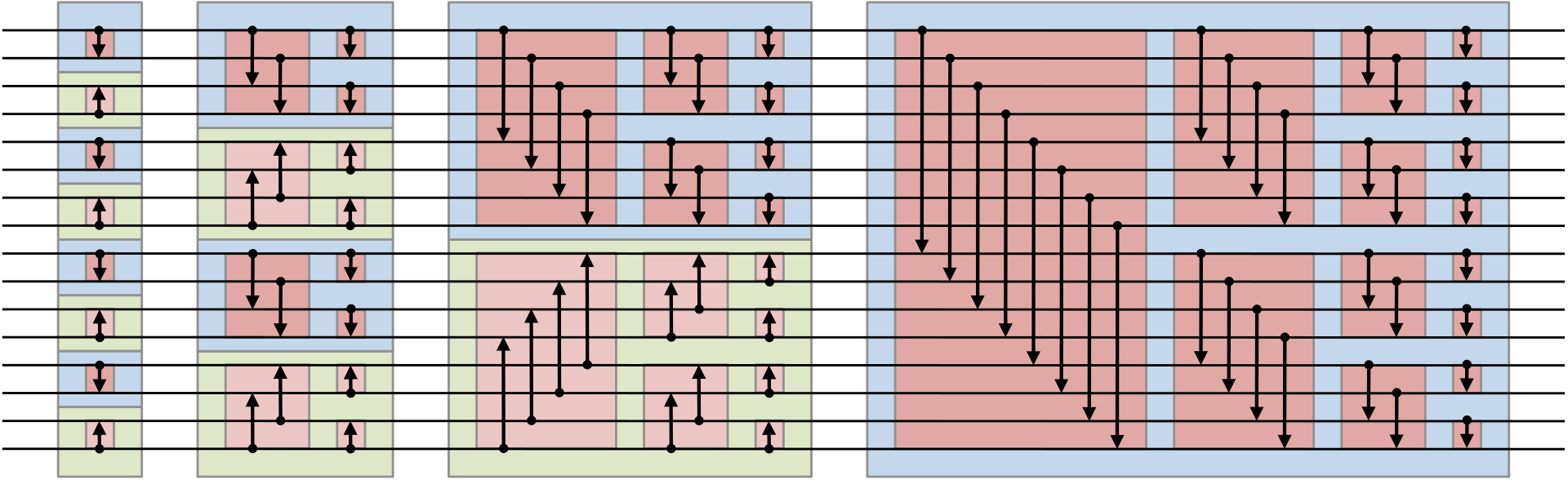}
\end{center}
\caption{A bitonic sorting network~\cite{bitonicsort,Bitonic-wiki} for $n = 16$.}
\label{fig:bitonic}
\end{figure*}

The rest of this subsection describes how to achieve Theorem~\ref{thm:bitonic}.

\subsection{Bitonic Sort}
A bitonic sorter~\cite{bitonicsort} is a sequence of $\log n$
layers of bitonic merges, where each bitonic merge is a butterfly network.
Figure~\ref{fig:bitonic} depicts a bitonic sorting network for 16 inputs~\cite{Bitonic-wiki}. 
The $n = 16$ inputs come in from the left, and go out on the right.
Each vertical arrow 
represents a comparator that compares and 
possibly swaps a pair of elements. At the end 
of the comparator, the arrow should point to the larger element.
We can see that in the first layer, there are $n/2$ butterfly networks (i.e., bitonic merges)
each of size $2$.
In the second layer, there are $n/4$ butterfly networks each of size $4$, and so on.
In the final layer, there is a single butterfly 
network of size $n$.

We implement bitonic sorting with the recursive {\sc Bitonic-Sort} algorithm given below.
It uses a recursive algorithm {\sc Bitonic-Merge} for the bitonic merging step
(which corresponds to each butterfly network
in Figure~\ref{fig:bitonic}), and we shall explain how to efficiently realize
\textsc{Bitonic-Merge} in the binary fork-join
model later in Section~\ref{sec:bitonic-merge}.

\begin{mdframed}
\begin{center}
{\sc Bitonic-Sort}$(A, flag)$
\end{center}

{\bf Input:} An unsorted array $A$ on $n$ distinct elements, and a $flag$ bit which is 1 if $A$ is to be sorted in increasing order
and 0 for a sort in decreasing order. Assume $n$ is a power of 2.

\vspace{3pt}
{\bf Algorithm}:

if $n = 1$ then return $A$; else continue with the following.
\begin{itemize}[leftmargin=6mm,topsep=0pt]
\item Let $A_l$ and $A_r$ be the subarrays of $A$ containing the first  and last $n/2$ elements of $A$ respectively.

Fork

$~~A'_l \leftarrow$ {\sc Bitonic-Sort}$(A_l,  flag)$\\
$~~~A'_r \leftarrow$ {\sc Bitonic-Sort}$(A_r, \overline{flag})$

Join

\item let $A'$ be the concatenation of $A'_l$ and $A'_r$

return {\textsc{Bitonic-Merge}}$(A', flag)$
\end{itemize}

\end{mdframed}

To analyze the performance bounds of 
\textsc{Bitonic-Sort}, we use  
$W_{\rm bmerge}(m)$, $Q_{\rm bmerge}(m)$ and $T_{\rm bmerge}(m)$ to denote  
the work, caching complexity and span, respectively, of {\textsc{Bitonic-Merge}}  
on an input of size $m$.
These bounds are obtained 
in Section~\ref{sec:bitonic-merge}, where 
 we will show that 
$W_{\rm bmerge}(m) = O(m \log m)$,  $T_{\rm bmerge}(m) = O(\log m \cdot \log\log m)$, and
$Q_{\rm bmerge}(m) = O((m/B) \cdot \log_M n)$.
For now, 
let us 
analyze the bounds for 
\textsc{Bitonic-Sort} assuming that the above bounds hold
for the \textsc{Bitonic-Merge} subroutine.

We have:

\vspace{-.3in}
\begin{eqnarray*}
W(n) &=& 2 \cdot W(n/2) + W_{\rm bmerge}(n)  ~~~\mbox{ with $W(2) =1$}\\
Q(n) &=& 2 \cdot Q(n/2) + Q_{\rm bmerge}(n) ~~~\mbox{ with $Q(n) = n/B$ if $n<M$}\\
T(n) &=& T(n/2) + T_{\rm bmerge}(n) ~~~\mbox{ with $T(2) =1$}
\end{eqnarray*}

Plugging in the values for $W_{\rm bmerge}$,  $T_{\rm bmerge}$, and
$Q_{\rm bmerge}$ from Section~\ref{sec:bitonic-merge}, we obtain our bounds for 
{\sc Bitonic-Sort}:

$W(n) = O( n \log^2 n)$, $T(n) =  O(\log^2 n \cdot \log\log n)$, and
$Q(n) = O((n/B) \cdot \log_M n \cdot \log (n/M))$, 
as long as $n \geq M \geq B^2$.

If we count only the number of comparisons for $W(n)$ and $T(n)$, the constant factor in the $O(\cdot)$ is
1/2. In our algorithm, in order to achieve cache efficiency, we also have matrix transpositions for which we 
have data movement costs. The data movement  
cost is about $n \log^2 n$ for work and about $\log^2 n\cdot \log\log n$ for the span. For $Q(n)$
the cache miss cost is about $2n/B \cdot \log_M n \cdot \log (n/M)$ including memory accesses for bitonic sort comparisons
as well as for  the matrix transposes.

\subsection{Bitonic Merge}\label{sec:bitonic-merge}

Earlier in our paper, we essentially showed how to realize a butterfly network
efficiently in a cache-agnostic, binary fork-join model.
The same approach applies to each bitonic merge within bitonic sort. 
In Figure~\ref{fig:bitonic}, 
we can see that each bitonic merge is a butterfly network.
In comparison with our earlier \textsc{Rec-ORBA} algorithm,
here the butterfly network is in the reverse direction:
in the first layer of a $m$-size bitonic merge, each element in the input sequence is compared to the
element $m/2$ away from it in the sequence, and it is only in the last layer that adjacent elements
are compared.  
For this reason, in our algorithm 
${\textsc{Bitonic-Merge}}$ below, we first perform a transpose 
of the $\sqrt{m}$ partitions,  
and then
perform the first batch of recursive calls. We then perform another transpose
right after the first batch of recursive calls.  
This recursion structure we adopt is also similar to the FFT algorithm 
in Frigo et al.~\cite{frigo1999cache}.
As in {\sc Bitonic-Sort} we include a $flag$ bit in
${\textsc{Bitonic-Merge}}$ along with the input $I$ to
indicate whether the output is in increasing or decreasing order.

The algorithm ${\textsc{Bitonic-Merge}}$ below obliviously evaluates
a single bitonic merge  within bitonic sort 
in the cache-agnostic, binary fork-join paradigm.
For simplicity, in this section, we assume that the size
of the problem $m$ is a perfect square in every recursive call. 
If $m$ is not a perfect square, we can deal with it in a similar way as
the earlier {\sc Rec\text{-}Orba}
algorithm.

\begin{mdframed}
\begin{center}
${\textsc{Bitonic-Merge}}({\bf I}, flag)$
\end{center}
\noindent {\bf Input:}
The input array ${\bf I}$ contains $m$ elements where $m$ is a power of $2$.

\vspace{5pt}
\noindent {\bf Algorithm:}
If $m = 1$, just return ${\bf I}$. If $m = 2$, apply a single comparator to the pair of  
elements and return the elements in increasing order if  $flag=1$ and in decreasing order if $flag=0$.
Otherwise, perform the following:
\begin{enumerate}[topsep=1pt,itemsep=1pt]
\item 
Write the $m$ inputs as a $\sqrt{m}\times \sqrt{m}$ matrix 
where each group of consecutive $\sqrt{m}$ elements form a row. 
Perform a transposition on this matrix, and let 
$J_1, \ldots, J_{\sqrt{m}}$
be the rows of the transposed matrix. 

\item 
In parallel: For $i \in [\sqrt{m}]$, recursively call 
$$K_i := {\textsc{Bitonic-Merge}}(J_i, flag).$$

\item 
Imagine that each $K_i$ is a row of a $\sqrt{m} \times \sqrt{m}$ matrix.
Perform a transposition on this matrix.
Let $L_1, \ldots, L_{\sqrt{m}}$ be the rows of the transposed matrix.

\item 
In parallel: For $i \in [\sqrt{m}]$, recursively call 
$${\bf O}_i := {\textsc{Bitonic-Merge}}(L_i, flag).$$

\item Return ${\bf O}_1$, $\ldots$, ${\bf O}_{\sqrt{m}}$.
\end{enumerate}
\end{mdframed}

\paragraph{Correctness.}
We argue why the above algorithm ${\textsc{Bitonic-Merge}}$ implements
a reverse-butterfly network as shown in Figure~\ref{fig:bitonic}. 
To see this, it suffices to show that Steps 1, 2, and 3 
together evaluate the first 
$\sqrt{m}$ layers of the butterfly network.
Henceforth we number the $m$ input elements to the butterfly
$0, 1, \ldots, m-1$. 

The first step performs a matrix transposition on the input elements. 
After the matrix transposition, 
for $i \in [0..\sqrt{m}-1]$, the $i$-th element is followed
by the elements with indices $i + \sqrt{m}$, $i + 2\sqrt{m}, 
\ldots, i + (\sqrt{m}-1) \sqrt{m}$
and they reside in the same row.

For $j \in \sqrt{m}$, in the $j$-th layer of the butterfly network,
elements with indices $i$ satisfying the condition 
$(i \mod (m/2^{j-1})) < m/2^{j}$  
gets paired with the element $i' := i + m/2^{j}$.
After the first matrix transposition, 
$i$ and $i'$ must be in the same row, and they are distance  
$\sqrt{m}/2^{j}$ apart --- notice that this is exactly a smaller reverse-butterfly structure 
within the same row whose size is $\sqrt{m}$.

This shows that after Steps 1 and 2, every element 
performs the same set of compare-and-swaps with the correct  
elements stipulated in the reverse-butterfly structure.
Now, Step 3 does a reverse transposition and rearranges the output
into the same order as they appear in the original reverse-butterfly network.

\paragraph{Analysis.}
The algorithm ${\textsc{Bitonic-Merge}}$'s performance
is charaterized by the following recurrences 
where $Q_{\rm bmerge}(m)$ denotes the cache complexity 
on a problem of size $m$, $T_{\rm bmerge}(m)$ denotes the span, 
and $W_{\rm bmerge}(m)$ denotes the total work.
\[
\begin{array}{l}
\text{for } m > 2: 
W_{\rm bmerge}(m) = 2 \sqrt{m} \cdot W_{\rm bmerge}(\sqrt{m}) + O(n) \\[5pt]
\text{for } m > M: 
Q_{\rm bmerge}(m) = 2 \sqrt{m} \cdot Q_{\rm bmerge}(\sqrt{m}) + O(m/B) \\[5pt]
\text{for } m > 2: 
T_{\rm bmerge}(m) = 2 \cdot T_{\rm bmerge}(\sqrt{m}) + O(\log m)
\end{array}
\]
The base conditions are:
\[
\begin{array}{l}
\text{for } m \leq  2: 
W_{\rm bmerge}(m) =  O(1), \ \ T_{\rm bmerge}(m) = O(1) \\[5pt]
\text{for } m \leq M: 
Q_{\rm bmerge}(m) = m/B
\end{array}
\]

The recurrences solve to 
$W_{\rm bmerge}(n) = O(n\log n)$, 
$T_{\rm bmerge}(n) = O(\log n \log \log n)$, and 
$Q_{\rm bmerge}(n) = O((n/B)\log_M n)$ 
as long as $n \geq M \geq B^2$.


\section{Background on OPRAM Simulation}~\label{sec:chan-opram}
In Chan, Chung, and Shi~\cite{opramdepth}'s OPRAM simulation algorithm, there are 
$O(\log s)$ recursion depths, where each  
recursion depth stores metadata called position labels for the next one. 
For each recursion depth $i$, there is a binary tree containing $2^i$ leaves (called
the {\it ORAM trees}~\cite{asiacrypt11,circuitopram}).
The top $\log_2 p$ levels of each 
tree is grouped together into a {\it pool} of size $O(p)$, in this way, the binary tree
actually becomes $O(p)$ disjoint subtrees.
Since the elements and their position labels are stored in random subtrees 
along random paths, except with negligible probability,   
only slightly more than logarithmically many requests will hit the same subtree.

Chan et al.~\cite{opramdepth}'s scheme
relies on a few additional oblivious 
primitives called oblivious ``send-receive'', ``propagation'', and ``aggregation''. 
We have introduced send-receive earlier.
For oblivious aggregation and 
propagation, we shall define them 
in Appendix~\ref{sec:oblivbldgblock}, and show that in the cache-agnostic, binary fork-join
model, they can be realized within the scan bound.

During each PRAM step, the scheme in~\cite{opramdepth}  relies on oblivious sorting and oblivious propagation  
to 
1) perform some preprocessing of the batch of $p$ memory requests;
2) use oblivious routing and attempt to fetch the requested elements and their position labels 
from the pools if they exist in the pools;
and 
3) perform some pre-processing to needed to look up the trees 
(since the requested elements may not be in the pools).
At this point, they sequentially look through all the $\log s$ trees, where in each
tree, a single path from the root to some random leaf is visited, taking $O(\log \log s)$ 
parallel time per tree, and $O(\log s \log \log s)$ depth in total. All other
steps in the scheme can be accomplished in $O(\log s)$ PRAM depth and thus  
in Chan et al.~\cite{opramdepth}, the sequential ORAM tree lookup phase is the bottleneck
in depth.

After the aforementioned fetch phase, they then 
need to perform maintenance operations to maintain
the correctness of the data structure: first, the fetched
elements and their position labels must be removed from the ORAM trees. 
To perform the removal in parallel without causing
write conflicts, some coordination effort is necessary
and the coordination can be achieved through using oblivious sorting 
and oblivious aggregation.
Next, 
the algorithm performs a maintenance operation called ``evictions'', where selected elements
in the pool are evicted back into the ORAM trees.
Each of the $\log s$ trees will have exactly two paths touched during the maintenance phase, 
and oblivious
sorting and oblivious routing techniques are used to select appropriate elements
from the pools to evict back into the ORAM trees.
Finally, a pool clean up operation is performed to compress the pools' size 
by removing filler elements acquired during the above steps --- this can be accomplished
through oblivious sorting.
We refer the reader to Chan et al.~\cite{opramdepth}
for a full exposition of the techniques.


\section{Applications}\label{sec:app}
\label{sec:apps}
We describe various applications of our new sorting algorithm, including
list ranking, Euler tour and tree functions, tree contraction, connected components,
and minimum spanning forest. 
Some of our data-oblivious algorithms asymptotically outperform
the previous best known insecure algorithms (in the cache-agnostic, binary fork-join model), 
while the rest of our results
match the previous best known insecure results. 
 Euler tour and tree contraction were  also considered 
in data-oblivious algorithms~\cite{oblgraph00,oblgraph01,blantongraph}
but the earlier works are inherently sequential.

\subsection{List Ranking}

In the list ranking problem we are given an $n$-element linked list on the elements $\{1, 2, \ldots , n\}$.
The linked list is represented by the successor array $S[1..n]$, where $S[i] = j$ implies that element 
$j$ is the successor of element $i$. If $i$ is the last element in the linked list, then $S[i]$ is
set to 0.
Given the successor array $S[1..n]$, the list ranking problem asks for the rank array $R[1..n]$,
where $R[i]$ is the number of elements ahead of element $i$ in the linked list represented by
array $S$, i.e., its distance to the end of the linked list. 
In the weighted version of list ranking, each element has a value, and needs to compute the sum of the values at elements ahead
of it in the linked list. Both versions can be solved with essentially the same algorithm.

To realize list ranking obliviously, 
we first apply an oblivious random permutation to randomly permute the input array. 
Then, every entry in the permuted array calculates its new index 
through an all-prefix-sum computation --- we assume that each entry also carries
around its index in the original array.
Now, every entry in the permuted array requests
the permuted index of its successor through oblivious routing. 
At this point, we can apply a non-oblivious list ranking
algorithm~\cite{CR12}.
Finally,  
every entry in the permuted array routes the answer back to the original array
and this can be accomplished with oblivious routing.

We thus have the following theorem and we
match the state-of-the-art, non-oblivious, cache-agnostic, parallel 
list ranking algorithm~\cite{CR12}.

\begin{theorem}[List ranking]
Assume that $M \geq \log^{1+\epsilon} n$ and that $n \geq M \geq B^2$. 
Then, we can obliviously realize 
list ranking achieving span $O(\log^2 n \cdot \log \log n)$,
cache complexity $O((n/B) \log_M n)$, and
total work $O(n \log n)$.
Moreover, the algorithm is cache-agnostic.
\end{theorem}

\subsection{Euler Tour}
In the Euler tour problem we are given an unrooted tree $T$, where each edge is duplicated, to
represent the two edges corresponding to a forward and backward traversal of the edge in depth first traversal of the tree,
when rooted at some vertex. This multigraph representation of $T$ has an Euler tour
since every vertex has even degree. 
The goal is to create an Euler tour on the edges of $T$ by creating a linked list on these edges, where the successor of each edge is
the next edge in the Euler tour. 

In the standard PRAM  algorithms literature,  the input is assumed to be 
represented by
circular adjacency lists $Adj(v)$ where 
$v$ is a vertex in $T$. Specifically, $Adj(v)$ 
stores all the edges of the form $(v, *)$ in a circular list.
Also, for each edge $(u,v)$ in $Adj(u)$ it is assumed that there is a pointer
to its other copy $(v,u)$ in $Adj(v)$ (and vice versa).
For each edge $(u,v)$ in $T$, let $Adjsucc (u,v)$ be the successor of
edge $(u,v)$ in $Adj(u)$. 
It can be shown that if we assign the successor on the Euler tour of each edge $(x,y)$ 
as $\tau((x,y)) = Adjsucc(y, x)$, 
then the resulting circular linked list is an Euler tour of $T$.  This is a 
constant time parallel algorithm on an EREW PRAM and is an $O(\log n)$ span and linear work algorithm with binary forks and joins.

\paragraph{Oblivious and cache-agnostic algorithm.}
The above parallel algorithm is not cache-efficient since a cache miss can be incurred for computing $\tau((x,y))$ for each edge $(x,y)$.
We now adapt the Euler tour algorithm in~\cite{CR12} to  obtain a data-oblivious, cache-agnostic, binary fork-join realization. 

Suppose that the input is a list of edges in the tree.
Our algorithm does the following: 
\begin{enumerate}[itemsep=1pt,topsep=1pt,leftmargin=5mm]
\item Make a reverse of each edge, e.g.,
the reverse of $(u, v)$ is $(v, u)$.
\item 
Sort all edges by the 
first vertex such that every edge $(u, v)$ except the last one, by looking at
its right neighbor, 
knows its successor in the (logical) circular adjacency list $Adj(u)$.  

Use oblivious propagation 
such that the last edge in each vertex's circular adjacency list 
learns its successor in the circular list too.
\item 
Finally, use oblivious send-receive where every edge
$(u, v)$ requests $(v, u)$'s successor in $Adj(v)$.
\end{enumerate}

All steps in the above algorithm can be accomplished in the sorting bound.

\paragraph{Tree computations with Euler tour.}
Euler tour is a powerful primitive.  Once an Euler tour $\tau$ of $T$ is constructed, the
tree can be rooted at any vertex, and many basic tree properties such as 
preorder numbering, postorder numbering, depth, and number
of descendants in the rooted tree can be computing using list ranking  
on $\tau$ (with an incoming edge to the root removed
to make it a non-circular linked list).
We can therefore obtain oblivious realizations
of these tree computations too, by invoking our oblivious Euler tour
algorithm followed by oblivious list ranking.
The performance bounds will be dominated by the list ranking step.


\bibliographystyle{alpha}
\bibliography{refs,crypto,cache_obliv}

\newcommand{\etalchar}[1]{$^{#1}$}
\begin{thebibliography}{AKL{\etalchar{+}}20b}

\bibitem[ABB00]{bfork01}
Umut~A. Acar, Guy~E. Blelloch, and Robert~D. Blumofe.
\newblock The data locality of work stealing.
\newblock In {\em Proceedings of the Twelfth Annual ACM Symposium on Parallel
  Algorithms and Architectures (SPAA)}, page 1–12. Association for Computing
  Machinery, 2000.

\bibitem[ABP98]{bfork08}
Nimar~S. Arora, Robert~D. Blumofe, and C.~Greg Plaxton.
\newblock Thread scheduling for multiprogrammed multiprocessors.
\newblock In {\em Proceedings of the Tenth Annual ACM Symposium on Parallel
  Algorithms and Architectures (SPAA)}, page 119–129, New York, NY, USA,
  1998. Association for Computing Machinery.

\bibitem[ACD{\etalchar{+}}20]{AC+20}
Zafar Ahmed, Rezaul Chowdhury, Rathish Das, Pramod Ganapathi, Aaron Gregory,
  and Mohammad~Mahdi Javanmard.
\newblock Low-depth parallel algorithms for the binary-forking model without
  atomics.
\newblock {\em {\rm arXiv:2008:13292}}, 2020.

\bibitem[ACN{\etalchar{+}}20]{bucketosort}
Gilad Asharov, T.{-}H.~Hubert Chan, Kartik Nayak, Rafael Pass, Ling Ren, and
  Elaine Shi.
\newblock Bucket oblivious sort: An extremely simple oblivious sort.
\newblock In Martin Farach{-}Colton and Inge~Li G{\o}rtz, editors, {\em 3rd
  Symposium on Simplicity in Algorithms, SOSA@SODA}, pages 8--14. {SIAM}, 2020.

\bibitem[AGNS08]{AGNS08}
Lars Arge, Michael~T. Goodrich, Michael Nelson, and Nodari Sitchinava.
\newblock Fundamental parallel algorithms for private-chip multiprocessors.
\newblock In {\em Proceedings of the {Twenty}-third {Annual} {ACM} {Symposium}
  on {Parallelism} in {Algorithms} and {Architectures}}, {SPAA} '08, New York,
  NY, USA, 2008. ACM.

\bibitem[AKL{\etalchar{+}}20a]{optorama}
Gilad Asharov, Ilan Komargodski, Wei-Kai Lin, Kartik Nayak, Enoch Peserico, and
  Elaine Shi.
\newblock {OptORAMa}: Optimal {Oblivious} {RAM}.
\newblock In {\em Advances in Cryptology - {EUROCRYPT} 2020}, 2020.
\newblock See also: \url{https://eprint.iacr.org/2018/892}.

\bibitem[AKL{\etalchar{+}}20b]{paracompact}
Gilad Asharov, Ilan Komargodski, Wei-Kai Lin, Enoch Peserico, and Elaine Shi.
\newblock Oblivious parallel tight compaction.
\newblock In {\em Information-Theoretic Cryptography (ITC)}, 2020.

\bibitem[AKL{\etalchar{+}}20c]{optimuspram}
Gilad Asharov, Ilan Komargodski, Wei-Kai Lin, Enoch Peserico, and Elaine Shi.
\newblock Optimal {Oblivious} parallel {RAM}.
\newblock Manuscript in preparation, 2020.
\newblock Personal communication with Asharov et al.

\bibitem[AKS83]{aks}
M.~Ajtai, J.~Koml\'{o}s, and E.~Szemer{\'e}di.
\newblock An {O}(n log n) sorting network.
\newblock In {\em STOC}, 1983.

\bibitem[AV88]{AV88}
Alok Aggarwal and S.~Vitter, Jeffrey.
\newblock The {Input}/{Output} {Complexity} of {Sorting} and {Related}
  {Problems}.
\newblock {\em Commun. ACM}, 31(9):1116--1127, September 1988.

\bibitem[Bat68a]{batcher68}
K.~E. Batcher.
\newblock {S}orting {N}etworks and {T}heir {A}pplications.
\newblock AFIPS '68 (Spring), 1968.

\bibitem[Bat68b]{bitonicsort}
Kenneth~E. Batcher.
\newblock Sorting networks and their applications.
\newblock In {\em American Federation of Information Processing Societies:
  {AFIPS} Conference Proceedings: 1968 Spring Joint Computer Conference,
  Atlantic City, NJ, USA, 30 April - 2 May 1968}, pages 307--314, 1968.

\bibitem[BCG{\etalchar{+}}08]{bfork02}
Guy~E. Blelloch, Rezaul~A. Chowdhury, Phillip~B. Gibbons, Vijaya Ramachandran,
  Shimin Chen, and Michael Kozuch.
\newblock Provably good multicore cache performance for divide-and-conquer
  algorithms.
\newblock In {\em Proceedings of the Nineteenth Annual ACM-SIAM Symposium on
  Discrete Algorithms (SODA)}, page 501–510, USA, 2008. Society for
  Industrial and Applied Mathematics.

\bibitem[BCP15]{opram}
Elette Boyle, Kai-Min Chung, and Rafael Pass.
\newblock Oblivious parallel ram.
\newblock In {\em Theory of Cryptography Conference (TCC)}, 2015.

\bibitem[BCR{\etalchar{+}}11]{habanero}
Zoran Budimli\'{c}, Vincent Cav\'{e}, Raghavan Raman, Jun Shirako, Sa\u{g}nak
  Ta\c{s}undefinedrlar, Jisheng Zhao, and Vivek Sarkar.
\newblock The design and implementation of the habanero-java parallel
  programming language.
\newblock In {\em Proceedings of the ACM International Conference Companion on
  Object Oriented Programming Systems Languages and Applications Companion
  (OOPSLA)}, page 185–186, New York, NY, USA, 2011. Association for Computing
  Machinery.

\bibitem[BFGS11]{bfork06}
Guy~E. Blelloch, Jeremy~T. Fineman, Phillip~B. Gibbons, and Harsha~Vardhan
  Simhadri.
\newblock Scheduling irregular parallel computations on hierarchical caches.
\newblock In {\em Proceedings of the Twenty-Third Annual ACM Symposium on
  Parallelism in Algorithms and Architectures}, SPAA ’11, page 355–366, New
  York, NY, USA, 2011. Association for Computing Machinery.

\bibitem[BFGS20]{optimalbfork}
Guy~E. Blelloch, Jeremy~T. Fineman, Yan Gu, and Yihan Sun.
\newblock Optimal parallel algorithms in the binary-forking model.
\newblock In {\em ACM Symposium on Parallelism in Algorithms and Architectures
  (SPAA)}, 2020.

\bibitem[BG04]{bfork05}
Guy~E. Blelloch and Phillip~B. Gibbons.
\newblock Effectively sharing a cache among threads.
\newblock In {\em Proceedings of the Sixteenth Annual ACM Symposium on
  Parallelism in Algorithms and Architectures (SPAA)}, page 235–244, New
  York, NY, USA, 2004. Association for Computing Machinery.

\bibitem[BGM99]{bfork07}
Guy~E. Blelloch, Phillip~B. Gibbons, and Yossi Matias.
\newblock Provably efficient scheduling for languages with fine-grained
  parallelism.
\newblock {\em J. ACM}, 46(2):281–321, March 1999.

\bibitem[BGS10]{lowdepthco}
Guy~E. Blelloch, Phillip~B. Gibbons, and Harsha~Vardhan Simhadri.
\newblock Low depth cache-oblivious algorithms.
\newblock In {\em Proceedings of the Twenty-Second Annual ACM Symposium on
  Parallelism in Algorithms and Architectures (SPAA)}, page 189–199, New
  York, NY, USA, 2010. Association for Computing Machinery.

\bibitem[Bit]{Bitonic-wiki}
\url{https://en.wikipedia.org/wiki/Bitonic_sorter}.

\bibitem[BL93]{bfork03}
Robert~D. Blumofe and Charles~E. Leiserson.
\newblock Space-efficient scheduling of multithreaded computations.
\newblock In {\em Proceedings of the Twenty-Fifth Annual ACM Symposium on
  Theory of Computing (STOC)}, page 362–371, New York, NY, USA, 1993.
  Association for Computing Machinery.

\bibitem[BL99]{bfork04}
Robert~D. Blumofe and Charles~E. Leiserson.
\newblock Scheduling multithreaded computations by work stealing.
\newblock {\em J. ACM}, 46(5):720–748, September 1999.

\bibitem[BSA13]{blantongraph}
Marina Blanton, Aaron Steele, and Mehrdad Alisagari.
\newblock Data-oblivious graph algorithms for secure computation and
  outsourcing.
\newblock In {\em ASIA CCS}, 2013.

\bibitem[CCS17]{opramdepth}
T.-H.~Hubert Chan, Kai-Min Chung, and Elaine Shi.
\newblock On the depth of oblivious parallel ram.
\newblock In {\em Advances in Cryptology -- ASIACRYPT 2017}, pages 567--597.
  Springer International Publishing, 2017.

\bibitem[CGG{\etalchar{+}}95]{CG+95}
Yi-Jen Chiang, Michael~T. Goodrich, Edward~F. Grove, Roberto Tamassia,
  Darren~Erik Vengroff, and Jeffrey~Scott Vitter.
\newblock External-memory graph algorithms.
\newblock In {\em Proceedings of the Sixth Annual ACM-SIAM Symposium on
  Discrete Algorithms}, SODA ’95, page 139–149, USA, 1995. Society for
  Industrial and Applied Mathematics.

\bibitem[CGLS18]{cacheoblosort}
T-H.~Hubert Chan, Yue Guo, Wei-Kai Lin, and Elaine Shi.
\newblock Cache-oblivious and data-oblivious sorting and applications.
\newblock In {\em SODA}, 2018.

\bibitem[CGS{\etalchar{+}}05]{x10}
Philippe Charles, Christian Grothoff, Vijay Saraswat, Christopher Donawa, Allan
  Kielstra, Kemal Ebcioglu, Christoph von Praun, and Vivek Sarkar.
\newblock X10: An object-oriented approach to non-uniform cluster computing.
\newblock In {\em Proceedings of the 20th Annual ACM SIGPLAN Conference on
  Object-Oriented Programming, Systems, Languages, and Applications (OOPSLA)},
  page 519–538, New York, NY, USA, 2005. Association for Computing Machinery.

\bibitem[CLRS09]{clr}
Thomas~H. Cormen, Charles~E. Leiserson, Ronald~L. Rivest, and Clifford Stein.
\newblock {\em Introduction to Algorithms, Third Edition}.
\newblock The MIT Press, 3rd edition, 2009.

\bibitem[CR08]{CR08}
{Rezaul A.} Chowdhury and Vijaya Ramachandran.
\newblock Cache-efficient dynamic programming algorithms for multicores.
\newblock In {\em Proc. {ACM} {SPAA}}, 2008.

\bibitem[CR10a]{CR07}
{Rezaul A.} Chowdhury and {Vijaya} Ramachandran.
\newblock The cache-oblivious gaussian elimination paradigm: Theoretical
  framework, parallelization and experimental evaluation.
\newblock {\em Theory of Computing Systems}, 47(1):878--919, 2010.
\newblock Preliminary version in SPAA 2007.

\bibitem[CR10b]{vijayasort}
Richard Cole and Vijaya Ramachandran.
\newblock Resource oblivious sorting on multicores.
\newblock In Samson Abramsky, Cyril Gavoille, Claude Kirchner, Friedhelm~Meyer
  auf~der Heide, and Paul~G. Spirakis, editors, {\em Automata, Languages and
  Programming, 37th International Colloquium, {ICALP} 2010, Bordeaux, France,
  July 6-10, 2010, Proceedings, Part {I}}, volume 6198 of {\em Lecture Notes in
  Computer Science}, pages 226--237. Springer, 2010.

\bibitem[CR12a]{CR12}
Richard Cole and Vijaya Ramachandran.
\newblock Efficient resource oblivious algorithms for multicores with false
  sharing.
\newblock In {\em 26th {IEEE} International Parallel and Distributed Processing
  Symposium, {IPDPS} 2012, Shanghai, China, May 21-25, 2012}, pages 201--214.
  {IEEE} Computer Society, 2012.

\bibitem[CR12b]{CR12b}
Richard Cole and Vijaya Ramachandran.
\newblock Revisiting the cache miss analysis of multithreaded algorithms.
\newblock In {\em Proc. {LATIN}}, 2012.

\bibitem[CR13]{CR13}
Richard Cole and Vijaya Ramachandran.
\newblock Analysis of randomized work-stealing with false sharing.
\newblock In {\em Proc. IPDPS}, 2013.

\bibitem[CR17a]{CR17b}
Richard Cole and Vijaya Ramachandran.
\newblock Bounding cache miss costs of multithreaded computations under general
  schedulers.
\newblock In {\em Proc. ACM SPAA}, 2017.

\bibitem[CR17b]{CR17}
Richard Cole and Vijaya Ramachandran.
\newblock Resource oblivious sorting on multicores.
\newblock {\em ACM Trans. Parallel Comput.}, 3(4), March 2017.
\newblock Preliminary version in ICALP 2010.

\bibitem[CRSB13]{CRSB13}
Rezaul~Alam Chowdhury, Vijaya Ramachandran, Francesco Silvestri, and Brandon
  Blakeley.
\newblock Oblivious algorithms for multicores and networks of processors.
\newblock {\em J. Parallel Distributed Comput.}, 73(7):911--925, 2013.

\bibitem[CS17]{circuitopram}
T.{-}H.~Hubert Chan and Elaine Shi.
\newblock Circuit {OPRAM:} unifying statistically and computationally secure
  orams and oprams.
\newblock In {\em Theory of Cryptography Conference, ({TCC})}, 2017.

\bibitem[FLPR99]{frigo1999cache}
Matteo Frigo, Charles~E Leiserson, Harald Prokop, and Sridhar Ramachandran.
\newblock Cache-oblivious algorithms.
\newblock In {\em Foundations of Computer Science, 1999. 40th Annual Symposium
  on}, pages 285--297. IEEE, 1999.

\bibitem[FLR98]{cilk5}
Matteo Frigo, Charles~E. Leiserson, and Keith~H. Randall.
\newblock The implementation of the cilk-5 multithreaded language.
\newblock In {\em Proceedings of the ACM SIGPLAN 1998 Conference on Programming
  Language Design and Implementation (PLDI)}, page 212–223, New York, NY,
  USA, 1998. Association for Computing Machinery.

\bibitem[FS09]{FS09}
M.~Frigo and V.~Strumpen.
\newblock The cache complexity of multithreaded cache-oblivious algorithms.
\newblock {\em Theory of Computing Systems}, 45:203--233, 2009.

\bibitem[Git]{tbb}
Github.
\newblock \url{https://github.com/oneapi-src/oneTBB}.

\bibitem[GO96]{oram00}
Oded Goldreich and Rafail Ostrovsky.
\newblock Software protection and simulation on oblivious {RAM}s.
\newblock {\em J. ACM}, 1996.

\bibitem[Gol87]{oram10}
O.~Goldreich.
\newblock Towards a theory of software protection and simulation by oblivious
  {RAM}s.
\newblock In {\em STOC}, 1987.

\bibitem[Goo11]{Goodrich-spaa11}
Michael~T. Goodrich.
\newblock Data-oblivious external-memory algorithms for the compaction,
  selection, and sorting of outsourced data.
\newblock In {\em {SPAA} 2011: Proceedings of the 23rd Annual {ACM} Symposium
  on Parallelism in Algorithms and Architectures, San Jose, CA, USA, June 4-6,
  2011 (Co-located with {FCRC} 2011)}, pages 379--388. {ACM}, 2011.

\bibitem[Goo14]{zigzag}
Michael~T. Goodrich.
\newblock Zig-zag sort: A simple deterministic data-oblivious sorting algorithm
  running in {O(N Log N)} time.
\newblock In {\em STOC}, 2014.

\bibitem[GOT13]{oblgraph01}
Michael~T. Goodrich, Olga Ohrimenko, and Roberto Tamassia.
\newblock Graph drawing in the cloud: Privately visualizing relational data
  using small working storage.
\newblock In Walter Didimo and Maurizio Patrignani, editors, {\em Graph
  Drawing}, pages 43--54, Berlin, Heidelberg, 2013. Springer Berlin Heidelberg.

\bibitem[GS14]{oblgraph00}
Michael~T. Goodrich and Joseph~A. Simons.
\newblock Data-oblivious graph algorithms in outsourced external memory.
\newblock In {\em Combinatorial Optimization and Applications - 8th
  International Conference, {COCOA} 2014, Wailea, Maui, HI, USA, December
  19-21, 2014, Proceedings}, volume 8881 of {\em Lecture Notes in Computer
  Science}, pages 241--257. Springer, 2014.

\bibitem[IKK12]{accesspatternleak}
Mohammad Islam, Mehmet Kuzu, and Murat Kantarcioglu.
\newblock Access pattern disclosure on searchable encryption: Ramification,
  attack and mitigation.
\newblock In {\em Network and Distributed System Security Symposium (NDSS)},
  2012.

\bibitem[J{\'{a}}J92]{jajabook}
Joseph J{\'{a}}J{\'{a}}.
\newblock {\em An Introduction to Parallel Algorithms}.
\newblock Addison-Wesley, 1992.

\bibitem[JLS20]{opqkasper}
Zahra Jafargholi, Kasper~Green Larsen, and Mark Simkin.
\newblock Optimal oblivious priority queues and offline oblivious {RAM}.
\newblock In {\em ACM-SIAM SODA}, 2020.

\bibitem[KD88]{KD97}
S.~Rao Kosaraju and Arthur~L. Delcher.
\newblock Optimal parallel evaluation of tree-structured computations by
  raking.
\newblock In John~H. Reif, editor, {\em {VLSI} Algorithms and Architectures,
  3rd Aegean Workshop on Computing, {AWOC} 88, Corfu, Greece, June 28 - July 1,
  1988, Proceedings}, volume 319 of {\em Lecture Notes in Computer Science},
  pages 101--110. Springer, 1988.

\bibitem[KR90]{KR90}
Richard~M. Karp and Vijaya Ramachandran.
\newblock Parallel algorithms for shared-memory machines.
\newblock In {\em Handbook of Theoretical Computer Science, Volume {A:}
  Algorithms and Complexity}, pages 869--942. Elsevier and {MIT} Press, 1990.

\bibitem[LWN{\etalchar{+}}15]{oblivm}
Chang Liu, Xiao~Shaun Wang, Kartik Nayak, Yan Huang, and Elaine Shi.
\newblock Oblivm: A programming framework for secure computation.
\newblock In {\em Proceedings of the 2015 IEEE Symposium on Security and
  Privacy}, page 359–376, USA, 2015. IEEE Computer Society.

\bibitem[Mic]{tpl}
Microsoft.
\newblock
  \url{https://docs.microsoft.com/en-us/dotnet/standard/parallel-programming/task-parallel-library-tpl?redirectedfrom=MSDN}.

\bibitem[NWI{\etalchar{+}}15]{graphsc}
Kartik Nayak, Xiao~Shaun Wang, Stratis Ioannidis, Udi Weinsberg, Nina Taft, and
  Elaine Shi.
\newblock Graphsc: Parallel secure computation made easy.
\newblock In {\em 2015 {IEEE} Symposium on Security and Privacy, {SP} 2015, San
  Jose, CA, USA, May 17-21, 2015}, pages 377--394. {IEEE} Computer Society,
  2015.

\bibitem[Ora]{jforkjoin}
Oracle.
\newblock
  \url{https://docs.oracle.com/javase/tutorial/essential/concurrency/forkjoin.html}.

\bibitem[PR02]{PR02}
Seth Pettie and Vijaya Ramachandran.
\newblock A randomized time-work optimal parallel algorithm for finding a
  minimum spanning forest.
\newblock {\em SIAM J. Comput.}, 31(6):1879--1895, 2002.

\bibitem[SCSL11]{asiacrypt11}
Elaine Shi, T.-H.~Hubert Chan, Emil Stefanov, and Mingfei Li.
\newblock Oblivious {RAM} with {$O((\log N)^3)$} worst-case cost.
\newblock In {\em ASIACRYPT}, 2011.

\bibitem[Shi20]{pathoheap}
Elaine Shi.
\newblock Path oblivious heap: Optimal and practical oblivious priority queue.
\newblock In {\em 2020 {IEEE} Symposium on Security and Privacy, {SP} 2020, San
  Francisco, CA, USA, May 18-21, 2020}, pages 842--858. {IEEE}, 2020.

\bibitem[SV82]{vishkinCC}
Yossi Shiloach and Uzi Vishkin.
\newblock An {O}(log n) parallel connectivity algorithm.
\newblock {\em J. Algorithms}, 3(1):57--67, 1982.

\bibitem[XCP15]{controlchannel}
Yuanzhong Xu, Weidong Cui, and Marcus Peinado.
\newblock Controlled-channel attacks: Deterministic side channels for untrusted
  operating systems.
\newblock In {\em Proceedings of the 2015 IEEE Symposium on Security and
  Privacy}, pages 640--656, USA, 2015. IEEE Computer Society.

\end{thebibliography}


\end{document}